\documentclass[12pt]{amsart}
\usepackage{amssymb, amsmath}
\usepackage{graphics}
\usepackage{epsfig}
\usepackage{amsxtra}
\usepackage{color}

\numberwithin{equation}{section}

\newtheorem{thm}{Theorem}[section]

\newtheorem{lemma}[thm]{Lemma}

\newtheorem{prop}[thm]{Proposition}
\newtheorem{cor}[thm]{Corollary}

\theoremstyle{definition}

\newcommand{\defeq}{\stackrel{\rm{def}}{=}}

\renewcommand{\Im}{\operatorname{\rm Im}\nolimits}

\def \rank {\operatorname{rank}}

\def \supp {\operatorname{supp}}

\def \sgn {\operatorname{sgn}}

\def \supp {\operatorname{supp}}
\def \tr {\operatorname{tr}}

\def \mco {{\mathcal O}}

\def \mcr {{\mathcal R}}

\def \loc {\operatorname{loc}}
\def \bded {\operatorname{bded}}
\def \ext {\operatorname{ext}}

\def \muS {\mu_{S,\operatorname{max}}}

\def \ol {\overline{\lambda}}

\def \msca {\operatorname{m}_{sc}}

\def \restrict {\upharpoonright}

\def \mcd {{\mathcal D}}
\def \Real {{\mathbb R}}

\def \Sphere {\mathbb{S}}
\def \Complex {\mathbb{C}}
\def \Natural {{\mathbb N}}
\def \Sphere {{\mathbb S}}

\def \mcr{{\mathcal R}}

\def \Integers {{\mathbb Z}}
\def \R {{\mathbb R}}

\def \mch {\mathcal H}

\def \pot {V}

\title [Scattering resonances in even dimensions]
{Some remarks on resonances in even-dimensional Euclidean
scattering}

   \author { T.\ J.\  Christiansen and P.\ D.\ Hislop}
\thanks{T.J.C. partially supported by NSF grant DMS 1001156, P.D.H.
partially supported by NSF grant 1103104.}

\begin{document}

%\date{today}

\begin{abstract}
The purpose of this paper is to
 prove some results about quantum mechanical black box scattering in even dimensions $d \geq 2$.
We study the scattering matrix and
 prove some identities which hold for its meromorphic continuation onto $\Lambda$, the Riemann surface of the logarithm function.
%In particular, we correct a result of \cite{s-t}.
 We relate
 the multiplicities of the poles of the continued 
scattering matrix to the multiplicities of the poles of the resolvent.  
Moreover, 
we show that the poles of the scattering matrix on the $m$th sheet
of $\Lambda$ are related to the zeros of a scalar function
defined on the physical sheet.
This paper contains a number of results about ``pure imaginary'' resonances.  As an example,
in contrast with the odd-dimensional case, 
we show that in even dimensions
 there are no ``purely imaginary'' resonances on 
any sheet of $\Lambda$ for Schr\"odinger operators with potentials
$0 \leq V \in L_0^\infty (\R^d)$.  
%For non-positive potentials, 
%there are at most finitely-many imaginary resonances on each sheet
%of the Riemann surface and the number of resonances
%is uniformly bounded on each sheet by the number of negative eigenvalues of the Schr\"odinger operator.
\end{abstract}

\maketitle
\section{Introduction}

This paper presents several results about resonances in 
quantum mechanical black box 
Euclidean scattering in even 
dimensions.  There are several objects which naturally may be 
called resonances.  {\em Resolvent resonances} occur as poles of the meromorphic
continuation of the cut-off resolvent, while {\em scattering resonances} are
poles of the meromorphic continuation of the scattering matrix. 
In this setting both lie on $\Lambda,$ the logarithmic cover of 
$\Complex \setminus \{0\}$. We prove
an identity clarifying the 
relationship between these two.  Moreover, 
we show there is  a scalar function on the physical 
region, the zeros of which correspond to poles of the scattering matrix on the
$m$th sheet of $\Lambda$.  We show the absence of ``purely imaginary'' resonances
for certain classes of operators.  This extends results of Beale \cite{beale},
and is in sharp contrast with the odd-dimensional case.  We observe a small
correction to an oft-quoted identity of \cite{s-t} about symmetries
of the scattering matrix.  This correction is
important in the context of ``pure imaginary'' resonances in even dimensions.
%We also prove an identity, analogous to 
%those well-known in other contexts, relating poles of the scattering matrix
%and poles of the resolvent.  

Let $V\in L^{\infty}_{0}(\Real^d;\Complex)$, and let $\Delta\leq 0$ be the
Laplacian on $\Real^d$.  Set $P=-\Delta +V$, and, for $\Im \lambda >0$,
set $R(\lambda)=(P-\lambda^2)^{-1}$ be the resolvent, which is bounded as
an operator on $L^2(\Real^d)$ for all but finitely many $\lambda$ with $\Im \lambda>0$.  Let $\chi\in L^{\infty}_{0}(\Real^d)$ be one on 
the support of $V$.  It is well known that $\chi R(\lambda)\chi$ has
a meromorphic extension to $\Complex$ if $d$ is odd.  If $d$ is even, 
the extension is to $\Lambda$, the logarithmic cover of 
$\Complex \setminus \{ 0\}$.
 In the latter case, we identify the physical space, where 
$R(\lambda)$ is bounded on $L^2$ (except for finitely
many points), with the subset of $\Lambda$ defined by
$\Lambda_0\defeq\{ \lambda \in \Lambda: 0<\arg \lambda<\pi\}$.  
The poles of the meromorphic continuation of $\chi R(\lambda)\chi$ 
are called {\em (resolvent) resonances}. 
 Moreover, a similar extension and a similar definition can be made 
for many compactly supported perturbations of $-\Delta$ on $\Real^d$; see
the ``black box'' definition of Sj\"ostrand-Zworski \cite{sj-zw},
recalled here in Section \ref{s:correction}.  For example,
the class of operators for which one
can make this meromorphic continuation and
the subsequent  definition  of resonances 
includes the Dirichlet or Neumann Laplacian on  
$\Real^d \setminus \overline{\mco}$, where $\mco\subset \Real^d$ is
a bounded open set with smooth boundary.

We shall be particularly interested in the case of even $d$.
For $m\in \Integers$, we set 
$$\Lambda_m=\{ \lambda\in \Lambda: m\pi<\arg \lambda <(m+1)\pi\}.$$
Thus, with our convention, $\Lambda_0$ corresponds to the physical region in even dimensional
scattering.  We call a point $\lambda \in \Lambda$ ``pure imaginary'' if 
$\arg \lambda = \pi/2 +k\pi$ for some $k\in \Integers$.  An example of 
our results is the following.  
\begin{thm} \label{thm:fixedsign}
Let the dimension $d \geq 2$ be even.
Let $V \in L_0^\infty (\Real^d;\Real)$ be a potential with fixed sign, that is,
either $V \geq 0$ or $- V \geq 0$.
If $V \geq 0$, the nonnegative Schr\"odinger operator $H_V = - \Delta + V$
has no purely imaginary (resolvent) resonances on any sheet $\Lambda_m$, $m \in \Integers$. If $- V \geq 0$, suppose that the operator $H_V$
has $0 \leq N_V < \infty$ negative eigenvalues on $\Lambda_0$.
%$- \sigma_j$, for $j = 1, \ldots, N_V$, with $\sigma_j > 0$.
Then the lower-semi-bounded operator $H_V$ has at most $N_V$ purely imaginary 
(resolvent)
resonances on each sheet $\Lambda_m$, for $m \in \Integers \setminus\{0\}$.
\end{thm} 
This theorem is proved in Section \ref{s:fixedsign}.
We note that this is in sharp contrast with the odd-dimensional case.
In odd dimensions, Lax and Phillips \cite{l-p} proved
lower bounds on the number of purely imaginary resonances for
Dirichlet or Neumann obstacle scattering.  They noted
that their technique  applies to 
Schr\"odinger operators with strictly nonnegative compactly
supported potentials.  See also \cite{m-s} for a related result.
 This was extended by A.\ Vasy to compactly-supported, bounded potentials of fixed sign \cite{vasy}.  To be more precise, 
let $V\in L^{\infty}_0(\Real^d)$ be a fixed 
sign potential so that there is an $\epsilon>0$,
and some nontrivial ball $B\subset \Real^d$ so that $|V|\geq \epsilon 
\chi_B$, where 
$\chi_B$ is the characteristic function of the ball $B$.  Then \cite{l-p,vasy}
showed that
in {\it odd dimensions} $d \geq 3$, the Schr\"odinger operator 
$H_V = - \Delta + V$
has an infinite number of purely imaginary resonances on the nonphysical sheet.
In fact, they proved a qualitative lower bound.  The number of such 
poles $N_{\operatorname{im}} (r)$ with norm at most $r$ satisfies, for large $r$:
\begin{equation}\label{eq:lower-bd1}
N_{\operatorname{im}} (r) \geq c_V r^{d-1} 
\end{equation}
for a positive constant $c_V$.  Here and always the resonances are counted with
multiplicity.

The results of Lax and Phillips \cite{l-p} for 
obstacle scattering were extended to 
 certain Robin-type boundary conditions by Beale \cite{beale}.  
Beale also noted that for even-dimensional scattering there are no 
purely imaginary resonances on $\Lambda_{-1}$ (and hence on $\Lambda_1$)
for Dirichlet or Neumann obstacle scattering, and at most finitely many
for certain Robin-type boundary conditions. 
See Corollary \ref{c:obstaclepureimag} for a more precise statement in the 
even-dimensional case
and for our extension of these results.  We prove additional results on the absence of purely
imaginary resonances in Section \ref{s:poles}.

We begin this paper with 
Proposition \ref{p:correction},
 a somewhat subtle correction to an identity from \cite{s-t}.  
We include this because we are unaware of a reference in which the correct version is explicitly
stated and because the subtle distinction has important consequences for the existence or not
of purely imaginary resonances in even dimensions.  
In fact the correct version has been implicitly mentioned
in \cite{l-p, beale} in the context of purely imaginary resonances. 
Although most of this paper is about even-dimensional Euclidean
scattering, Proposition \ref{p:correction} is a result for  compactly
supported black
box perturbations of the Laplacian on any $\Real^d$ as long
as $d\geq 2$.

Another result of our note is Proposition \ref{p:togettophys}.  
It follows from
this proposition  that
in even dimensions
to study the poles of the scattering matrix on $\Lambda_m$, it suffices to study the zeros
of a function which is holomorphic on $\Lambda_0$ with the possible exception of at most finitely many poles there.
This function is related in an explicit
 way to the scattering matrix.  This is familiar
in the odd dimensional case, where the poles of the scattering matrix in the nonphysical half plane
are, with perhaps finitely many exceptions, 
determined by the zeros of the determinant of the scattering matrix
in the physical half plane.  

Theorem \ref{thm:prpsm} and its Corollary 
\ref{c:rpsp}  give a relationship between the poles of the resolvent
(``resolvent resonances'')  and
the poles of the scattering matrix (``scattering resonances'').  
Again, this relation is well known in 
the odd dimensional case and is known for a very limited subset of $\Lambda$
in the even dimensional case-- see Section \ref{s:rpsp} for references.  To
the best of our knowledge there is not a proof of this result
 in the literature which is valid
for all points in $\Lambda$.

There are a number of results on the distribution of resonances which are
not intimately tied to the parity of the dimension.  
At least some of these rely on complex scaling, and as a consequence 
can only say something about resonances ``near'' the physical half plane.  
``Near'' generally meaning in some sector, of opening no greater than 
$\pi/2$.  We make no attempt to survey such results, but merely mention as an example  results on 
the distribution of resonances ``near'' the physical plane for the Laplacian
on the exterior of a strictly convex obstacle, \cite{sj-zwco}.  

For questions about distribution of resonances further from the 
physical half-plane, the case of even dimensions has received far less
attention than the case of odd dimensions.  Exceptions include
\cite{intissar, vodeveven, vodev2} in which upper bounds on resonance-counting
functions are obtained, and \cite{chen, ch-hi2,SaB1,TSaB, Tang1} 
for lower bounds.  Two papers which focus on resonances on the sheets $\Lambda_{\pm 1}$ are \cite{beale} and
\cite{Zworski2}.
Results of Beale \cite{beale} for purely imaginary resonances 
are recalled in Section \ref{s:poles}.  The paper \cite{Zworski2}
proves a Poisson formula in even dimensions.

\section{The black box formalism and relations for the scattering matrix}
\label{s:correction}

In this section, we allow $d\geq 2$ to be either even or odd.  Here we assume
 $P$ is a compactly supported ``black box'' perturbation of the Laplacian on $\Real^d$
satisfying the conditions of \cite{sj-zw}, including that $P$ is 
 self-adjoint.  We remark that some of the results of this 
paper use the self-adjointness of $P$ in an essential way.

   We recall the black box assumptions below for the convenience of the reader.
Note that if $\mco \subset \Real^d$ 
is a bounded open set with smooth boundary $\partial \mco$, 
and $V\in L^{\infty}_{0}(\Real^d \setminus \mco;\Real)$, 
these hypothesis are satisfied by the operator 
$-\Delta +V$ on $\Real^d\setminus\overline{ \mco}$ with
Dirichlet or Neumann boundary conditions on $\partial \mco$.

In recalling the assumptions of \cite{sj-zw} we use similar notation.
  By a black box operator we 
mean an operator $P$ defined on a domain ${\mathcal D}\subset \mch$ satisfying the conditions below.   
Let $U\subset \Real^d$ be a bounded connected open set.
Let $\mch$ be a complex Hilbert space with orthogonal decomposition
$$\mch=\mch_{U}\oplus L^2(\Real^d \setminus U).$$
Following \cite{sj-zw}, we denote the corresponding orthogonal
projections by $u\mapsto u_{\restrict U} $ and
$u \mapsto u_{ \restrict \Real^d  \setminus U} .$  We assume that the 
operator $P:\mch\rightarrow \mch$ is semibounded below, self-adjoint with
domain  ${\mathcal D}\subset \mch$.  Furthermore, if  $u\in H^2(
\Real^d \setminus U)$ and 
$u$ vanishes near $U$,  then $u\in {\mathcal D}$; and conversely
${\mathcal D}_{| \Real^d \setminus U}\subset H^2(\Real^d \setminus
U)$.
The operator $P$ is  $-\Delta $ outside $U$: 
$$Pu_{|  \Real^d \setminus U} = -\Delta
u_{| \Real^d  \setminus U}
\; \mbox{for all }u\in { \mathcal D}$$
and $${\bf 1}_{U}(P+i)^{-1} \; \mbox {is compact}$$
%and 
%$${\bf 1}_{U}(Q+i)^{-m_0} $$
%is trace class for some finite $m_0$, 
where ${\bf 1}_{U}$ is the characteristic
function of $U$; that is, projection onto $\mch_U$.
%{\em note: I have deleted here the trace class requirement.  I don't think we need it here but I should
%check.  Also, I should check how these assumptions match with sj-zw.}

Let $\chi\in L^{\infty}_{0}(\Real^d)$.  Under these conditions on $P$, the cut off resolvent
$\chi (P-\lambda^2)^{-1}\chi$  defined on the 
physical sheet $0<\arg \lambda <\pi$ has a meromorphic continuation
to $\Complex$ if $d$ is odd and to $\Lambda$ if $d$ is even.  
While this is well-known for specific 
operators, in this generality we refer the reader
to, for example, \cite[Theorem 1.1]{sj-zw} or the 
proof of Proposition 4.1 of \cite{pe-zwsc}.

For future reference, we note that we shall use the notation
$$ \langle x \rangle ^s \mch  \defeq \mch_U \oplus \langle x \rangle ^s L^2(\Real^d \setminus U)$$
and similarly for $\langle x \rangle ^s {\mathcal D}\subset 
\langle x \rangle ^s\mch.$  Here $\langle x \rangle = (1+|x|^2)^{1/2}$.

We work with the scattering matrix $S(\lambda)$ associated to $P$; 
one explicit expression
for it is recalled in (\ref{eq:scatmatrix}) and another in 
Proposition \ref{p:pe-zw}.  On the positive 
real axis  $\{\arg \lambda =0\}$ 
 it is unitary, and it differs from the identity by a trace class operator.  Moreover, it 
 has a meromorphic continuation to the complex plane (for $d$ odd) or
to $\Lambda$ (for $d$ even), as follows from the meromorphic continuation of the cut-off resolvent
and the expression for the scattering matrix recalled in 
Proposition \ref{p:pe-zw} (see also \cite{s-t, l-peven}).

Following \cite{s-t} we write 
$$\overline{\lambda}=|\lambda|\exp(-i \arg \lambda)\; \text{for $\lambda \in \Lambda$}.$$
For $\lambda\in \Lambda$ the complex involution $\lambda \mapsto \overline{\lambda}$ takes 
$\Lambda_m$ to $\Lambda_{-m-1}$.  The unitarity of the scattering matrix for $\arg \lambda =0$ 
means that 
$$S^*(\overline{\lambda}) S(\lambda)=I$$
in any dimension.

In \cite[Theorem 1]{s-t}, the following identity is stated for the scattering matrix 
$S(\lambda)$ for a combination of an obstacle and potential perturbation of the 
Laplacian on $\Real^d$:
\begin{equation}\label{eq:incorrect}
S(\overline{\lambda})^*= \left\{ \begin{array}{ll}
S(-\lambda) & \text{if $d$ is odd}\\
2I-S(e^{i\pi} \lambda) & \text{if $d$ is even}.
\end{array}
\right.
\end{equation}
%Here, following \cite{s-t}
 %for even $d$ we write $\overline{\lambda} =|\lambda|\exp (-i\arg \lambda)$ for $\lambda \in \Lambda$.

These relations have been widely repeated by others, including the 
present authors.  However, it appears that there is a slight error 
in the identity.  In most or all of the 
cases where  (\ref{eq:incorrect}) rather than
(\ref{eq:correct}) has been stated, the difference between the two 
versions is unimportant to the subsequent discussion.  
We shall see in Section \ref{s:poles} that the difference is important
 to results on pure imaginary resonances, and 
that is why we include Proposition \ref{p:correction} here.  

Define $\mcr:C^{\infty}(\Sphere^{d-1})\rightarrow C^\infty (\Sphere^{d-1})$ by 
$(\mcr f)(\theta)
= f(-\theta)$.  We shall use the same notation for the continuous extension of 
$\mcr$ to $L^2(\Sphere^{d-1})$.

\begin{prop}\label{p:correction} For $P$ satisfying the black box conditions, 
the scattering matrix
$S(\lambda)$ satisfies
%The correct version of the identity (\ref{eq:incorrect}) is
\begin{equation}\label{eq:correct}
S(\overline{\lambda})^*= \left\{ \begin{array}{ll}
\mcr S(-\lambda) \mcr  & \text{if $d$ is odd}\\
2I-\mcr S(e^{i\pi} \lambda)\mcr  & \text{if $d$ is even}.
\end{array}
\right.
\end{equation}
\end{prop}
There are certainly many instances in the literature which 
are consistent with (\ref{eq:correct}) rather than (\ref{eq:incorrect}).  These include, for example,
\cite{yafaev} and the works \cite{l-p} and \cite{beale}, 
both  of the latter related to work in this paper, see 
Corollary \ref{c:obstaclepureimag}.
However, since we are unaware of an explicit reference for (\ref{eq:correct}) and because the
distinction between the two versions is important for our results, we include Proposition
\ref{p:correction} and a proof here.  

We note that these equalities show that for the operators $P$ we consider, in even dimensions $d$, 
if $\lambda_0\in \Lambda$ is
a pole of $S(\lambda)$ with 
 $m\pi <\arg \lambda_0 < (m+1)\pi$, then $\overline{\lambda_0}\; e^{i\pi} $ is a pole of $S(\lambda)$,
 and $-m\pi<\arg(  \overline{\lambda_0}\; e^{i\pi}) < (-m+1)\pi$.  Thus, poles of $P$ on $\Lambda_m$ are
 symmetric with poles on $\Lambda_{-m}$. This replaces the symmetry relation for the case of odd-dimensional $d$, which is more familiar: for odd $d$, $\lambda_0\in \Complex$ is a pole of the scattering matrix
 if and only if $-\overline{\lambda}_0$ is a pole of the scattering matrix.
%  Moreover, while we give a proof which follows along the lines
%of that of \cite{s-t}, we note that using Melrose's \cite{lrb}
% definition of the relative scattering matrix (the
%"usual" scattering matrix) via the absolute scattering matrix would provide an alternate, perhaps more
%transparent, proof of this proposition.

%This correction may be well known, but since we are unaware of a reference
%and because the distinction between (\ref{eq:correct}) and
%(\ref{eq:incorrect}) is important to our results, we include a 
%proof here.
%Before doing so, we note that at least in the odd-dimensional potential 
%scattering case, (\ref{eq:correct}) is consistent with the expression 
%\cite[???]{yaffaev} for the scattering matrix (unitary for positive $k$),
%while (\ref{eq:incorrect}) is not for general potentials $\pot$.  
%Moreover, the work \cite{l-p}, done primarily for obstacle scattering 
%in dimension $d=3$, makes the point that $\mcr(I-S(i \sigma))$, $\sigma>0$,
%is a self-adjoint operator.  This is consistent with 
%the odd-dimensional case of (\ref{eq:correct}) but not with (\ref{eq:incorrect})
%without further assumptions.

%Let $\pot \in L^{\infty}_{\comp}(\Real^d; \Real)$.
%Let $P$ be the operator $-\Delta +\pot$ with 
%Dirichlet or Neumann boundary conditions  on $\Real^d\setminus \mco$, where
%$\mco $ is a compact set with smooth boundary, and $\Real^d  \setminus \mco$
%is connected.  Set $\Gamma = \partial (\Real^d \setminus \mco)$.
%We also allow the possibility that $\mco = \emptyset$,
%in which case $P$ has no boundary conditions.

In order to prove the proposition, we recall some background.
Let $g\in \mch$ satisfy $g_{\restrict |x|>R}=0$ for some finite $R$.  Then for 
$\lambda\in \Real\setminus \{0\} $
  %Then if$g \in L^{\infty}_{\comp}(\Real^d \setminus \mco)$ 
there are  unique
$u_{\pm}\in \langle x \rangle ^{1/2+\epsilon} {\mathcal D} $ 
satisfying
\begin{align}
(P-\lambda^2) u_\pm & =g\; \\
\label{eq:radcond}
 u_\pm(x) & = 
e^{\pm i \lambda |x| }|x|^{-(d-1)/2}(\alpha_\pm (x/|x|) +o(1)) \; \text{as $|x|\rightarrow \infty$}
\end{align} 
 for some functions $\alpha_\pm \in C^{\infty}(\Sphere^{d-1})$.

Let $\omega \in \Sphere^{d-1}$ and let
$\psi \in C^{\infty}(\Real^d)$ be $0$ on $U$ and satisfy $1-\psi \in C_c^{\infty}(\Real^d)$.  
Applying the results recalled above, we see that there are unique $v_{\pm }(x,k,\omega) \in 
\langle x \rangle ^{1/2+\epsilon}{\mathcal D}$ and 
thus unique scattering amplitudes $s_{\pm}(\theta,\lambda,\omega)$ which satisfy  (Cf. 
\cite[Section 2]{s-t}; we use similar notation here.)
\begin{align}\label{eq:vpm}
(P-\lambda^2) \left[v_\pm(x,\lambda,\omega)+\psi(x) 
e^{i\lambda x \cdot \omega}\right]& =0  
%{\mathcal B} \left[v_\pm(x,k,\omega)+e^{ik x \cdot \omega}\right]& = 0 \; 
%\text{on $\Gamma$}
\end{align}
with
\begin{equation}
v_\pm(r\theta, \lambda,\omega)= 
e^{\pm i \lambda r}r^{(1-d)/2}\left[ s_{\pm}(\theta, \lambda, \omega)
+o(1)\right] \; \text{as $r\rightarrow \infty$}.
\end{equation}

Now \cite[Lemma 2.1]{s-t} states
\begin{align}\label{eq:stright1}
\overline{s}_+(\theta,\lambda,\omega) & =
 s_-(\theta, \overline{\lambda}, -\omega)\\
 s_+ (\theta, e^{i\pi} \lambda, \omega) & 
= s_-(\theta, \lambda, -\omega)\label{eq:stright2}\\
s_- (\theta, \lambda, \omega)& = s_-(\omega, \lambda, \theta). 
\label{eq:stright3}
\end{align}
 The
reader should note that we mean by 
 the notation $\overline{s}$  the usual complex
conjugate on $\Complex$.  While this is possibly confusing, an alternate
notation with $^*$, risks being confused with the adjoint of an operator.

Strictly speaking, \cite{s-t} proved 
(\ref{eq:stright1})  -(\ref{eq:stright3}) only for $P$
which are  $-\Delta+V$ in the exterior of a smooth, compact obstacle with 
Dirichlet or Neumann boundary conditions.  However, it is 
not difficult to see that their proof extends to self-adjoint $P$ satisfying the black box type
conditions we consider here.

The following lemma is analogous to \cite[2.7]{s-t}.  It seems that it 
is this identity in which \cite{s-t} made an error.
\begin{lemma}\label{l:correction}
With the notation as above,
$s_+(\theta, \lambda, \omega)= s_+(-\omega, \lambda, -\theta)$.
\end{lemma}
 \begin{proof}
We have
$$s_+(\theta,\lambda, \omega)=
 \overline{s}_-(\theta, \overline{\lambda},-\omega)
= \overline{s}_-(-\omega, \overline{\lambda},\theta)=
s_+(-\omega,\lambda,-\theta)$$
where we have used respectively
(\ref{eq:stright1}), (\ref{eq:stright3}), and (\ref{eq:stright1}).
\end{proof} 

In \cite[(2.7)]{s-t} it is claimed that 
$s_+(\theta, \lambda, \omega)= 
s_+(-\theta, \lambda, -\omega)$.  This is in general not true, as we
shall see.  Here we denote by $\tilde{\mcr}$ the operator which on $
L^2(\Real^d)$ is given by $(\tilde{\mcr}f)(x) = f(-x)$.  We shall
also denote by $\tilde{\mcr}\mco$ the set $\{x\in \Real^d : -x\in \mco\}$,
with a similar definition of $\tilde{\mcr}\partial \mco.$
\begin{lemma}\label{l:reflection}
 Let $\mco \subset \Real^d$ be an open bounded set  
 so that $\mco $ has smooth boundary $\partial \mco$, and 
let $V\in L^{\infty}_{0} (\Real^d;\Real)$.  Let $P$ denote $-\Delta +V$ on $\Real^d \setminus 
\overline{\mco}$
with Dirichlet (Neumann) boundary conditions.  Let $P \tilde{\mcr}$ denote
$-\Delta + V(-x) $ on $\Real^d \setminus (\tilde{\mcr}\overline{\mco})$ with
 Dirichlet (respectively, Neumann) boundary
condition.
Then $$s_{+,P\tilde{\mcr}}(\theta, \lambda,\omega)= s_{+,P}(-\theta, \lambda, -\omega)$$
where $s_{+,P\tilde{\mcr}}$ and $s_{+,P}$ denote the functions $s_+$ corresponding to $P\tilde{\mcr}$ 
and $P$,
respectively.
\end{lemma}
We understand in the statement of this lemma that we may take $\mco = \emptyset$, in which case there is 
no boundary condition.
\begin{proof} We use notation similar to that of (\ref{eq:vpm}).
  We shall 
add a subscript $P$ to $v_+$   writing
$v_{+,P}$  to denote its dependence on $P$.  

Note that $v_{+,P}(-x,\lambda, -\omega)$ 
satisfies
$$(P\tilde {\mcr}-\lambda^2) 
\left[v_{+, P}(-x,\lambda,-\omega)+\psi e^{i\lambda x \cdot \omega}
\right] =0 \; \text{in $\Real^d\setminus \tilde{\mcr}\overline{\mco}$ }
$$
with $v_{+, P}(-x)$ satisfying the Dirichlet (or Neumann) boundary condition on $\partial \mcr \mco$. 
Moreover, since $v_{+,P}(-x,\lambda,-\omega)$ satisfies a radiation condition
as in (\ref{eq:radcond}) with the $+$ sign, we have 
for $\lambda\in \Real \setminus\{0\}$ that $v_{+,P \mcr }(x,\lambda,\omega)=
v_{+,P}(-x,\lambda,-\omega)$ and
thus
$$s_{+,P \tilde{\mcr}}(\theta, \lambda, \omega)= 
s_{+,P}(-\theta, \lambda, -\omega).$$
\end{proof}

We continue with the notation of the previous lemma, and 
 show that \cite[(2.7)]{s-t}, that is,
$s_+(\theta,\lambda,\omega)=s_+(-\theta,\lambda, -\omega)$,
 cannot hold in general.   
The scattering matrix
at energy $\lambda$ is determined by $s_+(\theta, \lambda, \omega)$, see (\ref{eq:scatmatrix}).
Thus by uniqueness results of inverse scattering theory
if  either $V\equiv 0$ and $\Real^d \setminus \overline{\mco}$ is
connected (e.g. \cite[Theorem 5.6]{l-pbook}) {\em or} $\mco =\emptyset$
and $V\in C_c(\Real^d;\Real)$ (e.g. \cite{e-w, faddeev, saito}),
 {\em and} if  $s_{+,P\tilde{\mcr}}(\theta,
\lambda, \omega)= s_{+,P}(\theta, \lambda, \omega)$, for all 
$\lambda\in (0,\infty)$ and all
$\theta, \omega \in \Sphere^{d-1}$, then $P\tilde{\mcr}= P$ and $\tilde{\mcr}\mco =\mco$.
%This holds if and only if  $\pot(-x)\equiv \pot(x)$ and $\tilde{\mcr}\mco= \mco$.  
Thus if we use Lemma \ref{l:correction} we see that 
\cite[(2.7)]{s-t} is not true in general.  However, we have shown, if we temporarily assume
Proposition \ref{p:correction}, the following corollary.
\begin{cor} \label{c:even}
 Suppose $\tilde{\mcr}\mco=\mco$ and $V(-x)\equiv V(x)$, with $\mco$ and $V$
satisfying the conditions of Lemma \ref{l:reflection}.  In this case, if $d$ is odd, $S(\overline{\lambda})^*=S(-\lambda)$,
and if $d$ is even, $S(\overline{\lambda})^*= 2I-S(e^{i\pi}\lambda)$.
\end{cor}

Now we prove the proposition.  We return to omitting the subscript $P$ on $s_+$, as we
will be working with a fixed operator $P$.

\vspace{2mm}
\noindent
{\em Proof of Proposition \ref{p:correction}}. 
We have \cite[Section 2]{s-t}
\begin{equation}\label{eq:scatmatrix}
S(\lambda)h(\theta)= h(\theta)+ \left( \frac{i\lambda}{2\pi}\right)^{(d-1)/2}
\int h(\omega) \overline{s}_-(-\theta, \overline{\lambda}, \omega)  dS_\omega.
\end{equation}
Now let $\lambda>0$, that is, $\arg \lambda =0$.
From (\ref{eq:scatmatrix}), we see that for such $\lambda$ the kernel of 
$S^*(\overline{\lambda}) -I$ is given by
\begin{equation}\label{eq:first}
\left(  \frac{\lambda}{2\pi}\right)^{(d-1)/2} e^{-\pi i(d-1)/4} 
s_-(-\omega, \lambda, \theta).
\end{equation}
On the other hand, the kernel of $S(e^{i\pi } \lambda)-I $ is
\begin{align} \label{eq:second}
\left( \frac{ i e^{i \pi }\lambda}{2\pi}\right) ^{(d-1)/2} 
\overline{s}_-(-\theta, e^{-i\pi }\overline{\lambda}, \omega) & = 
  \left( \frac{  \lambda}{2\pi}\right) ^{(d-1)/2}e^{3\pi i (d-1)/4}
   s_+(-\theta, e^{i\pi }\lambda, -\omega) \nonumber  \\
   & =  \left( \frac{  \lambda}{2\pi}\right) ^{(d-1)/2}e^{3\pi i (d-1)/4}
   s_-(-\theta, \lambda, \omega) 
   \end{align}
   from (\ref{eq:stright1}) and  (\ref{eq:stright2}).  Applying (\ref{eq:stright3}) to (\ref{eq:second}) and 
   comparing (\ref{eq:first}) finishes the proof of the 
   proposition.
   \qed

\section{Preliminary results on multiplicities of poles and some consequences of (\ref{eq:correct})} In this section we 
work only in 
even dimension $d$.  The main points of this section are to define the 
multiplicities of
poles of the resolvent and scattering matrix, and to prove
 Proposition \ref{p:togettophys} which
identifies poles of the scattering matrix on $\Lambda_{m+1}$ with zeros 
of a function defined
on $\Lambda_0$.

\subsection{Multiplicities of the poles of the resolvent}
This subsection recalls a result on the structure of the resolvent at a pole
and defines two notions for the multiplicity of the pole of the resolvent.

A result which we shall need is the following lemma, which
is essentially \cite[Lemma 2.4]{gu-zw} in a different setting.
We do not give a proof, as it follows essentially identically
the proof of that result.  We use notation similar to \cite[Lemma 2.4]{gu-zw}, 
but adapted for this context.  In the statement below and later in
this paper we use the notation $\lambda^2$ for the analytic continuation 
of the function which is equal to $\lambda^2$
for  $\lambda \in \Lambda_0
\simeq \{ z\in \Complex: \Im z>0\}$.
\begin{lemma}\label{l:gu-zw} (cf. \cite[Lemma 2.4]{gu-zw})
If $R$ has a pole at $\lambda_0\in \Lambda$, then there is a finite $p>0$
so that
$$R(\lambda)=\sum_{k=1}^p \frac{A_k(\lambda_0)}
{(\lambda^2-\lambda_0^2)^k} +H_R(\lambda_0,\lambda)$$
where $H_R(\lambda_0,\lambda)$ is holomorphic near $\lambda_0$. 
 There is a constant  $0<q<\infty$ 
so that
\begin{equation}
A_k(\lambda_0)= \sum_{l,m=1}^q a_k^{lm}
\varphi_l \otimes \varphi_m,
\end{equation}
with
$$\varphi_l \otimes \varphi_m (f) =\varphi_l \int f \varphi_m,\; $$
for $f\in \mch$ having $f \mapsto f_{ \restrict \Real^d  \setminus U}$
with compact support.
Moreover,  $\varphi_m$, $m=1,...q$ satisfy
\begin{align*} \varphi_m & \in \mch \; \text{if $\lambda_0 \in 
\partial \Lambda_0$}\\
\varphi_m & \in \mcd_{U}\oplus e^{|x|(|\lambda_0||\sin \arg \lambda_0|+\epsilon)} 
C^{\infty}_b(\Real^d \setminus U)
\; \text{otherwise}.
\end{align*}
If $a_k(\lambda_0) $ denotes the matrix
 $(a_{k}^{lm}(\lambda_0))_{1\leq l,m\leq q}$, then $a_1(\lambda_0)$ is
 symmetric with rank $q$, $d(\lambda_0)=a_{1}(\lambda_0)^{-1} a_2(\lambda_0)$ is nilpotent,
and $a_k(\lambda_0)=a_1(\lambda_0)d(\lambda_0)^{k-1}$, $k>1$.
\end{lemma}

We shall in fact need two notions related to the multiplicity 
of a pole of the resolvent.
We first define the multiplicity $\mu_R$  of a pole of the resolvent $R$ as follows.  Given $\lambda_0\in \Lambda$, define $\gamma_{\lambda_0}$ to be 
a small circle centered at $\lambda_0$ that contains no poles of 
the resolvent except, possibly, a pole at $\lambda_0$.
%If $R$ is holomorphic in a 
%neighborhood of $\lambda_0\in \Lambda$, then $\mu_R(\lambda_0)=0$.  If $\lambda_0$ is a pole of $R$, choose a circle $\gamma_{\lambda_0}$ 
%centered at $\lambda_0$ sufficiently small that it encloses no other singularities of $R$, and 
Define
$$\mu_R(\lambda_0) \defeq \rank  \int_{\gamma_{\lambda_0}} 
 R(\lambda) 2 \lambda d\lambda = \rank \int_{\gamma_{\lambda_0}}
 R(\lambda) d\lambda.$$
We note that by an argument just as in the proof of \cite[Lemma 2.4]{gu-zw}, this is the same as the dimension of the image of 
the singular part of $R$ at $\lambda_0$.  
%We shall call $\lambda_0$ a simple resonance if $\mu_R(\lambda_0)=1$.

We need another, more restrictive, notion of multiplicity related to the 
resolvent of $P$ satisfying the black box conditions. 
The need for this is related to 
the possibility of eigenvalues of $P_{\restrict \mch_U}$; we provide
an example below.

 Let
\begin{equation}
\label{eq:chirestrict}
\chi \in C_c^{\infty}(\Real^d) \; 
\text{satisfy $\chi\equiv 1$ on $\overline{U}$}
\end{equation}
where $U$ is as in the black box assumptions on $P$ of Section \ref{s:correction}.  Then define
$$\mu_{(1-\chi)R}(\lambda_0) \defeq \rank \int_{\gamma_{\lambda_0}}(1-\chi)R(\lambda)  d\lambda.$$
For any $\chi$, $\tilde{\chi}$ both satisfying (\ref{eq:chirestrict}),
 unique continuation together with the expansion of Lemma
\ref{l:gu-zw}  implies that
$\mu_{ (1-\chi)R}(\lambda_0) = \mu_{(1-\tilde{\chi})R}(\lambda_0).$

It is clear that 
\begin{equation}\label{eq:multineq}
\mu_{ (1-\chi)R}(\lambda_0) \leq \mu_{ R}(\lambda_0).
\end{equation}
Moreover, the inequality can be strict, and it is the strictness of this
inequality that makes having two definitions useful.  Consider the 
following example.  Let $\mco\subset \Real^d$ be an open bounded set with 
smooth boundary $\partial \mco$.
  Suppose in addition that $\Real^d \setminus \overline{\mco}$ has two 
connected components: $ \Real^d \setminus \overline{\mco}= \Omega_{\ext} 
\sqcup \Omega_{\bded}$, where $\Omega_{\ext}$ is unbounded and $\Omega_{\bded}$
is bounded, and each is connected.  An example of such an 
$\mco$ is an annulus in $\Real^2$. Then let $P= -\Delta$ with 
Dirichlet boundary conditions on $\Real^d \setminus \overline{\mco}$.  This operator
$P$ satisfies all the black box conditions.  It is really the direct sum of
two operators: one with discrete spectrum (the Dirichlet Laplacian on 
$\Omega_{\bded}$) and one with absolutely continuous spectrum
 (the Dirichlet Laplacian on $\Omega_{\ext} $).  The inequality
(\ref{eq:multineq}) is strict at points   $e^{i\pi m}\lambda_1$,
where $m\in \Integers$ and $\lambda_1^2$ is an eigenvalue of $P$.

We include the following lemma now to further explain the relationship
between the two notions of the multiplicity of a pole of the resolvent.
The proof uses a result of \cite{pe-zwsc}, a representation of 
the scattering matrix, recalled here in 
Proposition \ref{p:pe-zw}.  It also uses Lemma \ref{l:almostonedirection}, the proof
of which does not use the lemma below.
\begin{lemma}\label{l:diffresmult}
Suppose $\chi\in C_c^{\infty}(\Real^d)$ satisfies (\ref{eq:chirestrict}).
For $\lambda_0\in \Lambda$, 
$$\mu_R(\lambda_0)= \mu_{(1-\chi)R}(\lambda_0)+ 
\dim\{ f\in \mch: (P-\lambda_0^2)f=0\; \text{and}\; (1-\chi)f\equiv 0\}.$$
\end{lemma}
\begin{proof}
We use the notation of Lemma \ref{l:gu-zw}.
% It is not hard to
%see (for example, see the proof of \cite[Lemma 2.4]{gu-zw}) that
%\begin{equation}\label{eq:null}
%(P-\lambda_0^2)^{q+1}\varphi_m =0 \; \text{for $m=1,2,...,q$.}
%\end{equation}  
 By taking linear combinations
of the $\varphi_m$ if necessary and relabeling, we can assume that
there is an $n \in \{ 1,2, ...,q +1\}$ so that $(1-\chi)\varphi_m\equiv 0$ for 
$m=1,2,...,n-1$ and so that $(1-\chi)\varphi_n,...,(1-\chi)\varphi_q $ are 
linearly independent.  This $n$ is uniquely determined, and 
$n= 1+\mu_R(\lambda_0)- \mu_{R(1-\chi)}(\lambda_0)$.
%is $1$
%when $\mu_R(\lambda_0)= \mu_{R(1-\chi)}(\lambda_0)$.  Moreover,
%$q-n= \mu_R(\lambda_0)-\mu_{R(1-\chi)}(\lambda_0)$.
If $n>1$ then since $(P-\lambda_0^2) A_k(\lambda_0)=A_{k+1}(\lambda_0) 
$ and $ A_{p+1}(\lambda_0)=0$
(see \cite[(2.22)]{gu-zw}),  we must have at least one 
eigenfunction of $P$ in the span of $\{ \varphi_1,...,\varphi_{n-1}\}$,
and must have $\lambda_0^2 \in \Real$.
Suppose there is an $l_0\in \{ 1,2,...,n-1\} $ so that
$\varphi_{l_0}$ is {\em not} in the null
space of $P-\lambda_0^2$.  Since $(P-\lambda_0^2)\varphi_{l_0}$ is in 
the span of $\{ \varphi_1,...,\varphi_{n-1}\}$, it is in $\mch$. 
% Moreover,
%it is not hard to see (for example, see 
%the proof of \cite[Lemma 2.4]{gu-zw}) that 
But since $(P-\lambda_0^2)^p\varphi_m=0$ for $m=1,...,q$,  using
that $P$ is a self-adjoint and $\varphi_l \in \mch$
 this is a contradiction.  Thus
$$\mu_R(\lambda_0)\leq   \mu_{(1-\chi)R}(\lambda_0)+ 
\dim\{ f\in \mch: (P-\lambda_0^2)f=0\; \text{and}\; (1-\chi)f\equiv 0\}.$$
To see that 
 $$\mu_R(\lambda_0)\geq  \mu_{(1-\chi)R}(\lambda_0)+ 
\dim\{ f\in \mch: (P-\lambda_0^2)f=0\; \text{and}\; (1-\chi)f\equiv 0\}$$
one can use the expression for $S(\lambda)$ from Proposition \ref{p:pe-zw}
along with Lemma \ref{l:almostonedirection}.
\end{proof}

Thus we can see that it is the poles of $(1-\chi)R(\lambda)$ 
which are traditionally called resonances.
Often the poles $\lambda_0$ with $\arg \lambda_0/\pi \in \Integers$ are excluded, as they 
correspond to embedded eigenvalues.

\subsection{Multiplicities of the poles of the scattering matrix}

This subsection contains a number notions related to the multiplicity of
poles of the scattering matrix.  One such is measure of the multiplicities
of the zeros and poles of a scalar function, which we shall denote
$\msca$, with the ''sc'' for scalar.
Let $f$ be a scalar function meromorphic on $\Lambda$, not identically $0$,
and let $\lambda_o\in\Lambda$.  If $f(\lambda_0)=0$, define $\msca (f,\lambda_0)$ to be the 
multiplicity of $\lambda_0$ as a zero of $f$.  If $f$  has a pole at $\lambda_0$, define $\msca (f,\lambda_0)$ to be 
minus the order of the pole of $f$ at $\lambda_0$.  If $\lambda_0$ is neither a pole nor a zero of $f$, 
set $\msca (f,\lambda_0)=0$.  Thus $\msca(f, \cdot)$ is positive at zeros and negative at poles.

%Next we define the (maximum) multiplicity $\muS (\lambda_0)$
%of $\lambda_0\in \Lambda$ as a pole of $S$.  This requires some background.
%Recall that $S(\lambfa)-I$ is a trace class operator (away from the 
%poles of $S(\lambda)$ and that the singular part of $S$ at each pole is 
%finite rank.  

%We use some notation of Gohberg-Sigal, \cite{g-s}; see also \cite{b-p}.  Let
%$A(\zeta)$ denote an meromorphic family of Fredholm
 %operators defined on a Banach 
%space $\mathcal{B}$.  Assume there is a neighborhood $U$ of $\zeta_0$ so that
%$A$ is holomorphic on $U\setminus \{\zeta_0\}$, and that $A$ has 
%a meromorphic inverse $B$.  

Next we define the (maximum) multiplicity $\muS (\lambda_0)$ of $\lambda_0\in \Lambda$ as a pole of $S$.  Near $\lambda_0$, we may for some finite $p$ 
write $$S(\lambda) =\sum_{j=1}^p \frac{B_j(\lambda_0)}{(\lambda -\lambda _0)^j} + H_S(\lambda_0,\lambda)\defeq S_s(\lambda_0, \lambda)+
H_S(\lambda_0,\lambda)$$
where $H_S(\lambda_0,\lambda)$ is holomorphic near $\lambda_0$.
  Note that $B_j(\lambda_0) $ is finite rank for $j=1,...,p$, and since the 
$B_j$ are uniquely determined, so is $S_s(\lambda_0,\lambda)$.  Then set
\begin{equation}\label{eq:ourSmult}
\muS (\lambda_0)\defeq - \msca (\det( I+ S_s(\lambda_0,\lambda)), \lambda_0).
\end{equation}
We discuss this definition further
in Lemma \ref{l:wellknown}.
We note that our definition of the multiplicity of a pole of $S(\lambda)$ differs from one commonly used in scattering theory,
that is 
\begin{equation}\label{eq:othersmult}
- \frac{1}{2\pi i} \tr\left(  \int_{\gamma_{\lambda_0}} S^{-1}(\lambda) S'(\lambda) d\lambda \right) 
\end{equation}
where 
$\gamma_{\lambda_0}$ is a small circle centered at $\lambda_0$ and enclosing no singularities of $S$
or $S^{-1}$ except possibly $\lambda_0$
(see, for example, \cite[Equation 1.3]{bo-pe}).  Roughly speaking, the expression in (\ref{eq:othersmult}) counts the multiplicity
of the  
pole of $S$ at $\lambda_0$ minus the multiplicity of the  zero of $S$ at $\lambda_0$;
see Lemma \ref{l:wellknown}.  For many applications in scattering
theory this is sufficient, as one expects all but a finite number of the poles in one half plane of $\Complex$, 
and all but a finite number of 
zeros in the other half plane of $\Complex$.  
The even dimensional Euclidean scattering case is more complicated.  If $m\in \Integers $
has $|m|>1$, we expect in general that $\Lambda_m$ contains both infinitely many poles and infinitely many
 zeros of $S$.  Thus the definition (\ref{eq:ourSmult})
we use here counts the multiplicities of the poles without subtracting the multiplicities of the zeros.

The following lemma is well known, using that $S(\lambda )= (S^*(\overline{\lambda}))^{-1}$.   We outline a proof, in part in an effort to make notions
of multiplicities of a pole of the scattering matrix more transparent.
\begin{lemma} \label{l:wellknown}  Let $P$ be a self-adjoint operator satisfying the black box conditions
recalled in Section \ref{s:correction}.  Let $\gamma_{\lambda_0}$ be a small, positively 
oriented curve enclosing $\lambda_0$ and no zeros nor poles of
$S(\lambda)$, except possibly at $\lambda_0$ itself.  Then
$$ \frac{1}{2\pi i} \tr \int_{\gamma_{\lambda_0}}
S'(\lambda)S^{-1}(\lambda) d\lambda =  \muS(\overline{\lambda}_0)-\muS(\lambda_0)=
\msca( \det S(\lambda), \lambda_0).$$
%$$\mu_z(S(k),k_0)-\mu_p(S(k),k_0)=m (\det S(k), k_0).$$
\end{lemma}
\begin{proof}
We note that $S(\lambda)-I$ is a compact operator which is finitely meromorphic
on $\Lambda$;
that is, the only singularities of 
$S(\lambda)$ are poles, and at each pole the singular part is 
of finite rank.\footnote{Note that $\Lambda$ does not include any points which
project to the origin on the boundary of the physical half plane.} From 
\cite[Theorem 3.1]{g-s}, one can write near $\lambda = \lambda_0$,
$$ S(\lambda)= E(\lambda) D(\lambda) F(\lambda)$$
where $E,\; F$ are holomorphic with holomorphic inverses for $\lambda$ in 
a neighborhood of $\lambda_0$.  Moreover,
\begin{equation}\label{eq:D}
D(\lambda)= P_0 +\sum_{j=1}^n (\lambda-\lambda_0)^{k_j} P_j
\end{equation}
and the $P_j$, $j=0,\;1,...n$ are mutually orthogonal projections,
$\tr P_j=1$ for $j\geq 1$, $k_1,...,k_n$ are integers, and $Q\defeq 
I-\sum_{j=0}^nP_j$
is finite dimensional.  In fact, using that $S(\lambda)$ has a meromorphic
inverse, $Q=0$.  Moreover, the set $\{ k_1,...,k_n\}$ is uniquely 
determined by $S$ and $\lambda_0$.  Now
$$\muS(\lambda_0) = \sum_{j=1}^{n} \max(0,-k_j)$$
where the $k_j$ are as in (\ref{eq:D}).  
A comparison with  \cite{g-s} shows that this is what
is called  $P(S(\lambda_0))$ there.   In the notation of \cite{g-s}
$$N(S(\lambda_0))= \sum_{j=1}^n \max(k_j,0)$$
and in ours, using that $S(\lambda)^{-1}=
S^*(\overline{\lambda})$,
$$\muS(\overline{\lambda_0})= \sum_{j=1}^n \max(k_j,0).$$
Then from \cite[Theorem 2.1]{g-s}
we have
$$\frac{1}{2\pi i} \tr \int_{\gamma_{\lambda_0}}
S'(\lambda)S^{-1}(\lambda) d\lambda= \sum_{j=1}^n k_j
= \muS(\overline{\lambda_0})- \muS(\lambda_0).$$
%where $\gamma_{\lambda_0}$ is a small, positively oriented circle enclosing
%$\lambda_0$ and no other poles or zeros of $S$. 
 Finally, the second equality of the lemma is a special case of
\cite[Theorem 5.1]{g-s}.
\end{proof}

\subsection{A relation between poles of the scattering matrix
on $\Lambda_{m+1}$ and zeros of a scalar function on $\Lambda_0$}
Proposition \ref{p:togettophys} is proved in this section by fairly 
algebraic techniques.

The following lemma is a consequence of
(\ref{eq:correct}) and $S^*(\overline{\lambda})=S^{-1}(\lambda)$.
\begin{lemma}\label{l:productformula}
Let $d$ be even, $\lambda\in \Lambda$, and $m\in \Natural_0$.  Then
$$\left(S(\lambda)\mcr\right)
 \left( S(e^{i\pi}\lambda) \mcr \right) \cdot\cdot \cdot
\left( S(e^{im\pi} \lambda)\mcr \right)  
= \left[ (m+1)S(\lambda)-mI\right] \mcr^{m+1}.
$$
Moreover, 
$$\left( S( e^{im\pi}\lambda)\mcr \right) 
\left( S(e^{i(m-1)\pi}\lambda)) \mcr \right) \cdot \cdot \cdot 
\left( S(\lambda)\mcr \right) = \mcr^m \left[ (m+1)S(\lambda)-mI\right] \mcr.$$
\end{lemma}
\begin{proof}
The first identity trivially holds for $m=0$.  We now assume it holds for 
all integers between $0$ and $m$ inclusive, and show it holds for $m+1$.
From (\ref{eq:correct}), 
$$\mcr S(e^{i(m+1)\pi}\lambda) \mcr = 2I - S^*(e^{-i m \pi} \overline{\lambda})
= 2I -\left( S(e^{im\pi}\lambda)\right)^{-1}.$$
Multiplying both sides on the left by $S(e^{im\pi}\lambda)$ gives
\begin{equation}\label{eq:prodintermediate}
S(e^{im\pi}\lambda)\mcr S(e^{i(m+1)\pi}\lambda) \mcr = 2 S(e^{im\pi}\lambda)- I.
\end{equation}
We note that if $m=0$ this is the desired identity for $m+1=1$. So
assume $m\geq 1$.  Using 
the inductive hypothesis, multiply both sides of (\ref{eq:prodintermediate})
on the left by 
$$\left( S(\lambda)\mcr \right)\cdot \cdot \cdot 
\left( S(e^{i(m-1)\pi}\lambda)\mcr \right)
= [ m S(\lambda)-(m-1)I]\mcr^m$$
to obtain 
\begin{align*}
& \left(S(\lambda)\mcr\right) \left( S(e^{i\pi}\lambda) \mcr \right)
 \cdot\cdot \cdot
\left( S(e^{i(m+1)\pi} \lambda)\mcr \right)  \\
& = 2[ S(\lambda)\mcr \cdot \cdot \cdot S(e^{i(m-1)\pi }\lambda) 
\mcr S(e^{im \pi}\lambda)\mcr]
\mcr -[mS(\lambda)-(m-1)I]\mcr^m
\end{align*}
Using the inductive hypothesis again, 
we find
\begin{align*}
& \left(S(\lambda)\mcr\right) \left( S(e^{i\pi}\lambda)
 \mcr \right) \cdot\cdot \cdot
\left( S(e^{i(m+1)\pi} \lambda)\mcr \right) \\
& = 2 [ (m+1)S(\lambda)-mI]\mcr^{1+m+1}-[mS(\lambda )-(m-1)I]\mcr^m\\
& = [ (m+2)S(\lambda)-(m+1)I] \mcr^m
\end{align*}
as desired.

The proof of the second equality is very similar.
\end{proof}

\begin{prop}\label{p:togettophys}
For $\lambda_0\in \Lambda$, $m\in \Natural_0$,
\begin{align*}
\msca (\det( (m+1)S(\lambda)-mI), \lambda_0) & 
= \sum_{j=0}^m \msca(\det S(\lambda), e^{ij\pi} \lambda_0)\\
&  =\muS ( e^{i\pi (m+1)}\lambda_0)
- \muS ( \lambda_0).
\end{align*}
\end{prop}
Before proving the proposition, we note that it shows that the 
poles of the scattering matrix on $\Lambda_{m+1}$ correspond (with perhaps
a finite number of exceptions) to the zeros of a scalar function $\det ((m+1)S(\lambda)-mI)$
 on $\Lambda_0$.
This function is meromorphic on $\Lambda_0$, with at most finitely many 
poles (corresponding to eigenvalues of $P$) there. 
%Moreover, this
%function is defined as a determinant of an operator related in a 
%simple way to the scattering matrix.  
This is of course
familiar in the odd-dimensional case, where it is well known, and
has been extensively used, that with at most finitely many exceptions
the poles of the scattering matrix in the nonphysical half plane correspond to
zeros of the determinant of the scattering matrix in the physical half plane
$\Lambda_0$.

We also note that using the symmetry of the poles of the scattering 
matrix which is 
implied by Proposition \ref{p:correction}, poles of the scattering 
matrix on $\Lambda_{-m}$, $m\in \Natural$, can be identified with zeros 
of a scalar function using Proposition \ref{p:togettophys}.
\begin{proof}
We give the proof for $m=2l$ even; the proof for odd $m$ is similar.
Multiply both sides of the first identity of Lemma \ref{l:productformula}
(with $m$ replaced by $2l$)
by $\mcr$ on the right and rearrange slightly to get
$$S(\lambda)\left(\mcr S(e^{i\pi }\lambda) \mcr \right) S(e^{i 2\pi }\lambda) \cdot \cdot \cdot
\left(\mcr S(e^{i(2l-1)\pi}\lambda)\mcr\right) S(e^{i2l\pi}\lambda)= (2l+1)S(\lambda)-2l I.
$$
Since $S(\lambda)$ differs from the identity by a trace class operator,
so do $\mcr S \mcr$ and $(2l+1)S-2lI$.  Thus we have
\begin{multline*}
%\label{eq:2l1}
\det(S(\lambda))\det( \mcr S(e^{i\pi }\lambda) \mcr) \cdot \cdot \cdot 
\det( \mcr S(e^{i(2l-1)\pi}\lambda)\mcr )\det (S(e^{i2l\pi}\lambda))= \\
\det( (2l+1)S(\lambda)-2l I).
\end{multline*}
Using that $\mcr^2=I$ and $\det(I+AB)=\det(I+BA)$ when $A$ is trace class
and $B$ bounded,
$$\det(S(\lambda))\det(  S(e^{i\pi }\lambda)) \cdot \cdot \cdot 
\det( S(e^{i(2l-1)\pi}\lambda) )\det (S(e^{i2l\pi}\lambda))=
\det( (2l+1)S(\lambda)-2l I).
$$
This gives us
\begin{equation}\label{eq:sum}
\sum_{j=0}^{2l} \msca(\det (S(\lambda)), e^{i\pi j }\lambda_0)= 
\msca(\det\left( (2l+1)S(\lambda )-2l I\right), \lambda_0).
\end{equation}

By (\ref{eq:correct}), $\lambda_0\in \Lambda$ is a pole of $S(\lambda)$ if and only 
if $e^{\pi i}\overline{\lambda_0}$ is a pole of $S^*(\lambda)$, and the 
(maximum) multiplicities coincide.  
%Since $S^*(\overline{k})S(\lambda)=I$, $e^{\pi i}\overline{k_0}$ is a pole of $S^*(k)$
%if and only if $e^{-\pi i}k$ is a zero of $S(k)$.
%{\em worry about this!}.  
%Hence 
%$k_0\in \Lambda$ is a pole of $S(k)$ if and only if $e^{-\pi i}k$ is a zero
%of $S(k)$, and the multiplicities coincide. 
Applying this and Lemma \ref{l:wellknown},
\begin{align*}
\sum_{j=0}^{2l} \msca(\det (S(\lambda)), e^{i\pi j }\lambda_0) & = \sum _{j=0}^{2l}
\left( -\muS (e^{i\pi j} \lambda_0)+\muS (e^{-i\pi j}\overline{\lambda}_0)\right) \\ & 
 \sum _{j=0}^{2l}
\left( -\muS (e^{i\pi j} \lambda_0)+\muS (e^{i\pi (j+1)}\lambda_0)\right) \\&
= \muS (e^{i\pi(2l+1)}\lambda_0)-\muS (\lambda_0).
\end{align*}
Combined with (\ref{eq:sum}), this completes the proof.
% Applying this and
%the well known result that 
%$m(\det S(k),k_0)=\mu_z(S(k),k_0)-\mu_p(S(k),k_0)$ 
%{\em we may want instead to define $\mu_p$ to be the rank of 
%the singular part-- does this help us?  I don't think we've fixed 
%this completely!!!} to 
%(\ref{eq:sum}) completes the proof.
\end{proof}

\section{Poles of the resolvent and poles of the scattering matrix} 
\label{s:rpsp} In this
section we work only in even dimensions $d$.  The main result of this 
section is Theorem \ref{thm:prpsm}, an identification between the
poles of the resolvent and the poles
of the scattering matrix.  
%and Proposition \ref{p:togettophys} which identifies poles of the scattering 
%matrix on $\Lambda_n$ with zeros of a function defined on $\Lambda_0$. 
 While analogs of this result
are well known both in odd dimensions and for points in $\Lambda_1$ and $\Lambda_{-1}$ (see e.g. \cite{jensen, nedelec,pe-zwsc, s-t}), we are 
unaware of a proof in the literature which includes the other sheets of $\Lambda$.

We shall use \cite[Proposition 2.1]{pe-zwsc} which we
recall here for the convenience of the reader.  We have changed
the notation to be consistent with the notation of this paper.  We 
remark that there are a number of similar representations of 
the scattering matrix in the literature; see, for example, 
\cite[Section 2]{pe-st} or \cite[Section 3]{Zworski2}.  We recall
that our Hilbert space $\mch$ has an orthogonal decomposition
$$\mch = \mch_U\oplus L^2(\Real^d\setminus U)$$
where $U\subset \Real^d$ is a bounded open  set.
\begin{prop}(\cite[Proposition 2.1]{pe-zwsc})\label{p:pe-zw}   For $\phi\in C_c^{\infty}(\Real^d)$,
 let us denote by 
$${\mathbb E}^{\phi}_{\pm}(\lambda) :L^2(\Real^d)\rightarrow L^2(\Sphere^{d-1})$$
the operator with the kernel $\phi(x) \exp(\pm i \lambda \langle x,\omega \rangle )$.  Let us 
choose $\chi_i\in C_c^{\infty}(\Real^d)$, $i=1,\; 2,\;3$, such that $\chi_i\equiv 1$ near $U$
and $\chi_{i+1}\equiv 1$ on $\supp \chi_i$.  

Then for $0<\arg \lambda <\pi$ we have $S(\lambda)= I+A(\lambda)$, where 
$$A(\lambda) = i \pi (2\pi)^{-d} \lambda^{(d-1)/2} {\mathbb E}^{\chi _3}_+ (\lambda) [\Delta, \chi_1]
R(\lambda) [\Delta, \chi_2] ^t{\mathbb E}^{\chi_3}_-(\lambda)$$
where $^t{\mathbb E}$ denotes the transpose of ${\mathbb E}$.  The identity
holds for $\lambda \in \Lambda$ by analytic continuation.
\end{prop}

For $\lambda>0$ let $\Phi(\lambda,x,\omega)$  be the function satisfying
$$(P-\lambda^2)\Phi(\lambda,x,\omega)=0$$
$$\Phi(\lambda,r\theta,\omega)=e^{-i\lambda r \theta\cdot \omega}
+ r^{-(d-1)/2}e^{i\lambda r} 
\left(
s_+(\theta,\lambda, -\omega)+ o(r)\right) \text{as $r\rightarrow \infty$}.$$
Here we understand that $P$ acts in the $x$ variable, and 
$r>0$, $\theta\in \Sphere^{d-1}$.The function $\Phi$   
 has a meromorphic extension to $\lambda\in \Lambda$ which we denote in
the same way.  Note that if $\chi_1\in C_c^{\infty}(\Real^d)$ satisfies
the conditions of Proposition \ref{p:pe-zw}, then 
$$\Phi(\lambda, x, \omega)= (1-\chi_1)e^{-i\lambda x\cdot \omega}
-R(\lambda) [\Delta, \chi_1]e^{-i\lambda x \cdot \omega}.$$
 We shall also denote by $\Phi(\lambda)$ the operator 
from $L^2(\Sphere^{d-1})$ to $ \mch_U\oplus L^2_{\loc}(\Real^d \setminus U )$
 which is given by 
$$(\Phi(\lambda)f)(x)=\int_{\Sphere^{d-1}} f(\omega) 
\Phi(\lambda,x,\omega) dS_\omega,$$
and by $\Phi^t(\lambda)$ the transpose.
By Stone's formula, for $\lambda >0$
\begin{equation}\label{eq:resdiff}
R(\lambda)-R(\lambda e^{i\pi})=
 \alpha_d\lambda^{d-2} \Phi(\lambda) \Phi^t(\lambda e^{i\pi}),
\end{equation}
 where $\alpha_d= -i(2\pi)^{1-d}/2 $; compare \cite[(2.26)]{lrb}. 
 Since both sides have meromorphic extensions to $\Lambda$, the equality holds for $\lambda \in \Lambda$
 away from the poles.

The next two lemmas pave the way for Lemma \ref{l:almostonedirection},
which expresses the resolvent at $e^{im\pi} \lambda$ in terms of 
$S(e^{im\pi} \lambda)$, $R(\lambda)$, and $\Phi(\lambda)$.
\begin{lemma}  \label{l:phisphi}
For $\lambda\in \Lambda$,
$$\Phi(\lambda e^{i\pi }) = \Phi(\lambda ) \mcr S^*(\overline{\lambda}).$$
\end{lemma}
\begin{proof}
Although this well known, we sketch 
the proof.  For $\lambda>0$, and $x\in \Real^n\setminus U$,
\begin{multline*}
\Phi(\lambda e^{i\pi}, x,\omega )-\left(\Phi(\lambda) 
\mcr S^*(\overline{\lambda})\right)
(x,\omega)\\ =
 e^{i\lambda|x|}|x|^{-(d-1)/2}( \beta(x/|x|, \omega)) + O(|x|^{-(d+1)/2})
\; \text{as $|x|\rightarrow \infty$}
\end{multline*}
for some function $\beta \in C^{\infty}(\Sphere^{d-1}\times \Sphere^{d-1})$.
By Rellich's uniqueness theorem, since $\Phi(\lambda 
e^{i\pi }) - \Phi(\lambda) \mcr S^*(\overline{\lambda})$ 
is in the null space of $P-\lambda^2$, this is enough to show the
difference is $0$.  The general result follows by analytic continuation.
\end{proof}

\begin{lemma}\label{l:phis} For $m \in \Natural$ and $\lambda \in \Lambda$,
\begin{multline*}
\Phi(\lambda e^{im \pi})\Phi^t(\lambda e^{i(m+1)\pi})\\
= \Phi(\lambda) \left[ (m+1)^2 \mcr^{m+1} S^*(\ol e^{-im\pi})\mcr^{m}
 -m^2 \mcr^m S^*(\ol e^{-i(m-1)\pi}) \mcr^{m-1} -2m\mcr\right] \Phi^t(\lambda).
%\Phi(k) 
%\left[ (m+1)^2 S^*(\overline{k} e^{-i m \pi}) -m^2 S^*(\overline{k} e^{-i (m-1)\pi}) - 2m \mcr
%\right] \Phi^t (k).
\end{multline*} 
\end{lemma}
\begin{proof} 
By repeatedly applying Lemma \ref{l:phisphi} and the identity 
$(\mcr S^*)^t= \mcr S^*$, we have
\begin{multline}\label{eq:rdiff1}
\Phi(\lambda e^{im \pi})\Phi^t (\lambda e^{i(m+1)\pi}) 
= \\
\Phi(\lambda) \mcr S^*(\overline{\lambda})
 \mcr S^* (\ol e^{-i\pi}) \cdot \cdot \cdot \mcr S^*(\ol e^{-i(m-1) \pi})
\mcr S^*(\ol e^{-im\pi}) \cdot \mcr S^{*}(\ol e^{-i(m-1)\pi})\cdot \cdot \cdot \mcr  S^{*}(\ol)
\Phi^t(\lambda). 
\end{multline}
Lemma \ref{l:productformula} implies that for $p\in \Natural_0$
$$\mcr S^*(e^{ip\pi}\lambda)\mcr S^*(\lambda e^{i(p-1)\pi})\cdot \cdot \cdot \mcr S^*(\lambda)= \mcr^{p+1}\left[ (p+1) 
S^*(\lambda)-pI\right] .$$
Applying this identity with $\lambda e^{ip \pi}$ replaced by $\ol $
and with $p =m$, we find
 that  (\ref{eq:rdiff1})
 is 
$$ \Phi(\lambda) \mcr^{m+1}\left[
(m+1)S^*(\ol e^{-im\pi})-mI \right] \mcr 
S^*(\ol e^{-i(m-1)\pi})\cdot \cdot \cdot \mcr S^*(\ol)
\Phi^t(k).$$
Distributing and then using the second part of Lemma \ref{l:productformula}
twice, this is
\begin{align*} & 
\Phi(\lambda)  \mcr^{m+1} (m+1)\left[
(m+1)S^*(\ol e^{-im\pi})-mI \right] \mcr^{m} \Phi^t(\lambda) \\
& \hspace{3mm}
-m  \Phi(\lambda) \mcr^m \left[mS^*(\ol e^{-i(m-1)\pi})-(m-1)I\right]\mcr^{m-1}
  \Phi^t(\lambda)\\ & 
= \Phi(\lambda) \left[ (m+1)^2 \mcr^{m+1} S^*(\ol e^{-im\pi})\mcr^{m}
 -m^2 \mcr^m S^*(\ol e^{-i(m-1)\pi}) \mcr^{m-1} -2m\mcr\right] \Phi^t(\lambda).
\end{align*}
\end{proof}

The next lemma allows us to express the resolvent on
$\Lambda _m$, $m\in \Natural$, in terms of the resolvent on $\Lambda_0$,
the generalized eigenfunctions $\Phi(\lambda)$ on $\Lambda_0$, 
and the scattering matrix $S$ on $\Lambda_m$.
\begin{lemma}\label{l:almostonedirection}
 Let $P$  satisfy the general black box conditions recalled in
 Section \ref{s:correction}.  
%Suppose
%$k_0\in  \Lambda_m$, $m\in \Integers$, and $(e^{-i\pi m}k)^2$ is not an eigenvalue of $P$.
Then
for $m\in \Natural$,
\begin{equation}\label{eq:summary}
R(e^{im\pi}\lambda)-R(\lambda)= 
%\alpha_d k^{d-2} \Phi(\lambda) \left[
%m^2\mcr^m ( 2I-\mcr S(e^{im\pi}\lambda)\mcr)\mcr^{m-1}-m(m-1)\mcr\right] 
%\Phi^t(\lambda).
\alpha_d m\lambda^{d-2} \Phi(\lambda) \mcr^{m+1}\left[ m S(e^{im\pi}\lambda)
 -(m+1)I \right]\mcr^m \Phi^t(\lambda)
\end{equation}
\end{lemma}
\begin{proof}
We have
\begin{align*} 
R(e^{im\pi } \lambda)-R(\lambda) & = 
\sum_{j=1}^m (R(e^{i j\pi}\lambda)- R(e^{i (j-1)\pi}\lambda))\\ 
& = - \sum_{j=1}^m \alpha_d \lambda^{d-2} \Phi(e^{i (j-1)\pi}
\lambda ) \Phi^t(\lambda e^{ij\pi}) \\
& = - \sum_{j=1}^m \alpha_d \lambda^{d-2} 
\Phi(\lambda) \left[ j^2 \mcr^{j} S^*(\overline{\lambda}e^{-i(j-1)\pi})
\mcr^{j-1} \right. 
\\ & \hspace{8mm} \left. -
(j-1)^2\mcr^{j-1} S^*(\overline{\lambda}e^{-i(j-2)\pi})\mcr^{j-2} 
 - 2(j-1) \mcr\right] \Phi^t(\lambda )
\end{align*}
where the second equality follows from (\ref{eq:resdiff}) and 
the third follows from Lemma \ref{l:phis}.

Now 
\begin{multline*}\sum_{j=1}^m \left[j^2 \mcr^j 
S^*(\overline{\lambda} e^{-i(j-1)\pi})
\mcr^{j-1}
-(j-1)^2 \mcr^{j-1} S^*(\overline{\lambda}e^{-i(j-2)\pi}) \mcr^{j-2}-2
(j-1) \mcr \right] \\
= m^2 \mcr^{m} S^*(\overline{\lambda}e^{-i(m-1) \pi})\mcr^{m-1} -m(m-1)\mcr
\end{multline*}
using the fact that the first two summands telescope.  
Since
$$S^*(\overline{\lambda }e^{-i(m-1)\pi})= 2I-\mcr S(e^{i m\pi}\lambda )\mcr$$
from (\ref{eq:correct}), this proves the lemma.
\end{proof}

We now turn more directly to the central result of this section.
\begin{thm}\label{thm:prpsm} Let
$d$ be even and
 $P$  satisfy the 
 general black box conditions recalled in Section \ref{s:correction},
and let $\chi\in C_c^{\infty}(\Real^d)$ have $\chi\equiv 1$ on $\overline{U}$.
Then for $\lambda_0\in \Lambda$,
$$\mu_{(1-\chi)R}(\lambda_0) - \mu_{(1-\chi)R}(\overline{\lambda_0})
=-\msca(\det S(\lambda),\lambda_0) = \muS(\lambda_0)-\muS(\overline{\lambda_0})
$$
and
$$ 
\mu_{R}(\lambda_0) - \mu_{R}(\overline{\lambda_0})
=-\msca(\det S(\lambda),\lambda_0)=\muS(\lambda_0)-\muS(\overline{\lambda_0}) .$$
\end{thm}
We note the second equality in each displayed equation follows from Lemma
\ref{l:wellknown}.

This result is well-known in odd dimensions, and in even dimensions is known for
 $\lambda \in \Lambda_1 
\cup \Lambda_{-1}$, \cite{jensen, nedelec, pe-zwsc, s-t}.  
As we are unaware of a proof in the literature valid
 for other points in $\Lambda$ 
for even dimensions $d$, we include it in this section. 
 The proof we shall give of Theorem \ref{thm:prpsm} follows rather closely the proof of an analogous result
of Borthwick and Perry
for asymptotically hyperbolic manifolds given in \cite{bo-pe},
and a similar result (for a subset of $\Lambda$) in \cite{pe-zwsc}.  
The paper \cite{bo-pe} uses Agmon's perturbation
theory of resonances \cite{agmon} together with some ideas of Klopp and Zworski's paper \cite{kl-zw} to
prove a result on generic simplicity of resonances (away from certain points).  The analog of Theorem 
\ref{thm:prpsm} is then proved first for situations in which the resonances are simple, and then for the 
general case using as an ingredient the genericity result.

Denote by $Y^m_l$, $l=0,1,2,...$, $m=1,2,... m(l)$ a complete orthonormal set of
 the spherical harmonics
on $\Sphere^{d-1}$, where $m(l)= \frac{2l+d-2}{d-2}\binom{l+d-3}{d-3}$.  These eigenfunctions of the Laplacian $\Delta_{\Sphere^{d-1}}$
on $\Sphere^{d-1}$ satisfy
$$-\Delta _{\Sphere^{d-1}}Y^m_l=l(l+d-2)Y^m_l,\;
 l=0,\; 1,\;2,...\; m=1,\; 2,\;...,\; m(l).$$
From \cite[Lemma 3]{stefanov}
\footnote{The proof in \cite{stefanov} holds for $d\geq 3$ at least; 
the proof for $d=2$ follows from
the Jacobi-Anger expansion.}
\begin{equation}\label{eq:expexp}
e^{i \lambda x\cdot \omega}= (2\pi )^{d/2}\sum_{l=0}^\infty \sum_{m=1}^{m(l)} i^l \overline{Y}^m_l(\theta) 
Y^m _l(\omega)
( \lambda r)^{1-d/2}J_{l+d/2-1}(\lambda r),\; x=r\theta.
\end{equation} 
This equality is classical for $d=3$.

The next
 two lemmas prove special cases of  Theorem \ref{thm:prpsm}
 under an assumption of
simplicity of the pole of $(1-\chi)R$ at $\lambda_0$. 
% A similar, but more restrictive, statement occurs in \cite[??]{pe-zw}.

\begin{lemma}\label{l:simplep}
Let $P$ be self-adjoint and satisfy the other general black box conditions recalled in Section
\ref{s:correction} and let $\chi\in C_c^{\infty}(\Real^d)$, with 
$\chi \equiv 1 $ on $\overline{U}$. 
Let $\lambda_0\in \Lambda$ and suppose $\mu_{(1-\chi)R}(\lambda_0)\leq 1.$  Then $\mu_{(1-\chi)R}(\lambda_0)=\muS(\lambda_0)$. 
\end{lemma}
%\begin{proof} When $\mu_{R(1-\chi)}(\lambda_0) \leq 1$, we can see from 
% the 
%results of \cite[Proposition 2.1]{pe-zwsc}, recalled here in 
%Proposition \ref{p:pezw}, that $\mu_R(\lambda_0) \geq \muS (\lambda_0)$.  

%Now if $\lambda_0= e^{im \pi} \lambda_1$ for some $m\in \Natural$ and $\lambda_1\in 
%\overline{\Lambda_0}$,
%then  Lemma \ref{l:almostonedirection}
 % shows that  $\muS (\lambda_0) \geq \mu_R(\lambda_0)$ when 
%$\muS (\lambda_0) \leq 1$, thus finishing the proof when $\lambda_0\in \Lambda_m,$ $m \in \Natural$.

%If $\lambda_0= e^{im\pi}\lambda_1$, with $-m\in \Natural $ and 
%$\lambda_1\in \overline{\lambda_0}$, the result follows from the $m>0$ case
%using $R(k)= R^*(\overline{\lambda} e^{i\pi})$ and $S(\lambda)= 2I- \mcr S^*( e^{i\pi } \overline{\lambda})\mcr$.

%\end{proof}

\begin{proof}
When $\mu_{(1-\chi)R}(\lambda_0)\leq 1$, from Proposition \ref{p:pe-zw} we see that we must have $\muS(\lambda_0)\leq \mu_{(1-\chi)R}(\lambda_0)$.

So suppose  $\mu_{(1-\chi)R}(\lambda_0)=1$.  It follows from 
Lemma \ref{l:gu-zw}, Lemma \ref{l:diffresmult} and its proof 
that $R$ has a simple
pole at $\lambda_0$.   It is enough to show that $\muS(\lambda_0)\geq 1$.
From   the proof of \cite[Lemma 3.2]{sj-zw},  \cite[Equation 4.2]{pe-zwsc}
or
\cite[Equation 4.1]{vodeveven},
\begin{equation}\label{eq:rform}
(1-\chi) R(\lambda)= (1-\chi) R_0(\lambda)(I+K(\lambda))^{-1}
\end{equation}
where $R_0(\lambda)$
 is the resolvent for $-\Delta$ on $\Real^d$ and 
$K(\lambda): \mch_U\oplus L^2_{0}(\Real^d \setminus U)
\rightarrow \mch_U\oplus L^2_{0}(\Real^d \setminus U) $ is a compact operator.  Thus since $(1-\chi)R(\lambda)$ has a simple pole of rank $1$
at $\lambda_0$ by assumption, the residue of $(1-\chi)R$ at $\lambda_0$ is of the form (see Lemma \ref{l:gu-zw})
$$a (1-\chi)\varphi \otimes \varphi$$
where $a$ is a nonzero constant,  $(P-\lambda^2_0)\varphi=0$, 
$\varphi \not \equiv 0$, 
and $(1-\chi)\varphi\not \equiv 0$.  
Moreover, since by (\ref{eq:rform}) and the more explicit expressions
for $K(\lambda)$ found in the references given,
$(1-\chi)\varphi =(1-\chi) R_0(\lambda_0) g$ for some 
$g\in \mch_U\oplus L^2_{0}(\Real^d 
\setminus U)$, 
we have
\begin{equation}\label{eq:varphiexp}
(1-\chi)\varphi(r\theta)= \sum_{l,m} c_{lm} Y^m_l(\theta) r^{1-d/2} 
H^{(1)}_{l+d/2-1}(r\lambda_0 ),\; \text{for sufficiently large $r$}.
\end{equation}
By unique continuation not all of the $c_{lm}$ can be $0$.

To show that $\muS(\lambda_0)\geq 1$, by Proposition \ref{p:pe-zw} it is enough to show that 
$$C^{\infty}(\Sphere^{d-1})\ni \int \varphi(x) [\Delta, \chi_j]e^{i\lambda x\cdot \omega} dx \not \equiv 0, \; j=1,2$$
for $\chi_j$ as in the statement of the proposition.
 Using Green's
 theorem,
\begin{align}\label{eq:ibp}
\int \varphi(x) [\Delta, \chi_j]e^{i\lambda  x\cdot \omega } dx &  =
\int \varphi(x) [\Delta, \chi_j -1]e^{i\lambda  x\cdot \omega } dx \nonumber 
\\ & = 
 -
\int_{|x|=R}\left(  \varphi (x) \frac{\partial}{\partial |x|} e^{i\lambda x\cdot \omega }- 
e^{i \lambda x\cdot \omega }
 \frac{\partial}{\partial |x|}\varphi \right) dS
\end{align}
 if $R>R_0$, where $R_0$ satisfies $\supp \chi_j
\subset  B(0,R_0)=\{ x\in \Real^d: |x|<R_0\}$.
Notice the right hand side of (\ref{eq:ibp}) is  independent  $R$ satisfying this condition. 
Applying this with the expansions (\ref{eq:varphiexp}) and (\ref{eq:expexp})
 gives, for $R$ sufficiently large
\begin{equation} \label{eq:intexpsh}
\int \varphi(x) [\Delta, \chi_j]e^{i\lambda_0  x\cdot \omega } dx  
= \sum  Y^m_l(\omega)  c_{lm}(2\pi)^{d/2} i^l \lambda^{1-d/2}_0 
g_{lm}
\end{equation}
where 
\begin{multline*}
g_{lm}= 
R^{d/2} 
\left[J_{l+d/2-1}(\lambda_0 R)  \frac{\partial}{\partial R} \left( R^{1-d/2} H^{(1)}_{l+d/2-1} (\lambda_0 R)\right) \right. \\ \left.
-H  ^{(1)}_{l+d/2-1} (\lambda_0 R)\frac{\partial}{\partial R} \left( R^{1-d/2} J_{l+d/2-1}(\lambda_0 R) \right) \right].
\end{multline*}  Of course, $g_{lm} $ is in fact independent of $R$.

Now we use that if $\alpha \in \Natural$,  $J_\alpha(\lambda)$ and 
$H_\alpha (\lambda) $ are holomorphic functions on $\Lambda$ satisfying, for
$m\in \Integers,$
$$J_\alpha(e^{im\pi}\lambda )= e^{im\alpha \pi}J_\alpha (\lambda)$$
and 
$$H^{(1)}_\alpha (e^{im\pi}\lambda) = 
(-1)^{m\alpha}\left[ H^{(1)}_\alpha(\lambda)-2 m J_\alpha  (\lambda)\right]$$
\cite[pg.\ 239, 4.13]{Olver1}.  Using that $\lambda_0= e^{im\pi}\lambda_1$ for some 
$m \in \Integers$ and $\lambda_1$ with $0\leq \arg \lambda_1 <\arg \pi,$ 
\begin{equation} 
g_{lm}
=R \left[ J_{l+d/2-1}(\lambda_1 R)
\frac{\partial}{\partial R} H^{(1)}_{l+d/2-1}(\lambda_1 R) 
- H^{(1)}_{l+d/2-1}(\lambda_1 R)\frac{\partial}{\partial R} J_{l+d/2-1}(\lambda_1 R) \right].
\end{equation}

We recall the expansions, valid as $|z|\rightarrow \infty$,
\begin{equation}\label{eq:hankelexp}
 H^{(1)}_\nu (z) = \left(\frac{2}{\pi z}\right)^{1/2} e^{i(z -\pi \nu/2-\pi/4)}
\left( 1+ O(|z|^{-1}) \right), \; \text{if}\;  -\pi+\delta \leq \arg z \leq 2\pi -\delta
\end{equation}
and 
\begin{equation}\label{eq:besselexp}
J_\nu(z) =\left(\frac{1}{2\pi z}\right)^{1/2}
e^{-i(z-\nu \pi/-\pi/4)}\left( 1+O(1/|z|)\right) \; 
\text{if} \; \delta <\arg z< \pi  -\delta
\end{equation}
\cite[Equation 9.2.7]{Olverintorder} and \cite[Equation 9.2.1]{Olverintorder}.
In each case, $\delta>0$ and the principal branch of the square root is taken.
If we apply these expansions and the related one for the derivatives
(See  \cite[Equations 9.2.11 and 9.2.13]{Olverintorder};
the leading order terms can be obtained by differentiating the 
leading order terms of (\ref{eq:hankelexp}) and (\ref{eq:besselexp}).) we find
$$g_{lm}= 2i/\pi+O(1/R)\; \text{as} \; R\rightarrow \infty $$
provided $\arg \lambda_1 \not = 0$.  Since $g_{lm}$ is in fact constant,
 we have shown that $g_{lm}= 2i/\pi$. 
Thus by (\ref{eq:intexpsh}), $\int \varphi(x) [\Delta, \chi_j]e^{i\lambda_0  x\cdot \omega } dx  \not \equiv 0$ and $S$ has a simple pole with a residue of rank $1$ at
$\lambda_0$ if $\arg \lambda_1\not = 0$.

We finish our proof 
 by noting that our assumption that $\lambda_0$ is a simple pole
of $(1-\chi)R(\lambda)$ means that $\arg (\lambda _0 )/\pi \not \in \Integers$.
To see this, recall that under our assumption that $P$ is self-adjoint 
this is well known for $\arg \lambda_0 =0$.  Then Proposition \ref{p:togettophys},
combined with the fact that $S$ is unitary on $\arg \lambda =0$ since
$P$ is self-adjoint,
shows that $S(\lambda)$ cannot have poles with 
$(\arg \lambda)/\pi \in \Integers$.
Finally, (\ref{eq:summary}) and its adjoint equation show that 
$(1-\chi)R(\lambda)
 $
cannot have poles with  $(\arg \lambda)/\pi \in \Integers$.
\end{proof}
%Now if $\lambda_0= e^{im \pi}\lambda_1$, with $\lambda_1\in \overline{\Lambda_0%}$ and $m\in
%\Natural$,
%applying (\ref{eq:summary}) with $k= k_1$ implies the proposition,
%since $R(k)$ and $\Phi(k)$ will be holomorphic in a neighborhood of 
%$k_1$ by assumption.  

%Now if $k_0= e^{im \pi}k_1$, with $k_1\in \overline{\Lambda_0}$ and $m\in
%\Natural$,
%applying (\ref{eq:summary}) with $k= k_1$ implies the proposition,
%since $R(k)$ and $\Phi(k)$ will be holomorphic in a neighborhood of 
%$k_1$ by assumption.  If $k= e^{im\pi}k_1$, with $-m\in \Natural $ and 
%$k_1\in \overline{\lambda_1}$, the result follows from the $m>0$ case
%using $R(k)= R^*(\ok e^{i\pi})$ and $S(k)= 2I- \mcr S^*( e^{i\pi } \ok)\mcr$.

The next lemma builds a bit on the previous one.
\begin{lemma}\label{l:multdiff1} Let $P$  satisfy the
general black box conditions recalled in Section \ref{s:correction},
 and let $\chi\in C_c^{\infty}(\Real^d)$, with $\chi \equiv 1$ on $\overline{U}$.
Let $\lambda_0\in \Lambda$ and suppose both $\mu_{(1-\chi)R}(\lambda_0)\leq 1 $ 
and $\mu_{(1-\chi)R}(\overline{\lambda_0})\leq 1. $
Then
$$\mu_{R}(\lambda_0)- \mu_{R}(\overline{\lambda_0}) = -\msca(\det S(\lambda),
\lambda_0).$$
\end{lemma}
\begin{proof}
Applying Lemmas \ref{l:wellknown} and \ref{l:simplep} we see that
\begin{equation}\label{eq:lemmap}
\mu_{(1-\chi)R}(\lambda_0)- \mu_{(1-\chi)R}(\overline{\lambda_0}) = 
-m_{sc}(\det S(\lambda), \lambda_0).
\end{equation}
It follows from Lemmas \ref{l:diffresmult} and \ref{l:almostonedirection}
that for $m\in \Integers$
\begin{equation}\label{eq:diffequality}
\mu_R(\lambda e^{im\pi})- \mu_{(1-\chi)R}(\lambda e^{im\pi})= 
\mu_R(\lambda)- \mu_{(1-\chi)R}(\lambda).
\end{equation}
Next we note that if 
$\mu_R(\lambda)- \mu_{(1-\chi)R}(\lambda)>0$ then $2(\arg \lambda)/\pi \in \Integers$ by
Lemma \ref{l:diffresmult} and using the self-adjointness of $P$.  But if this holds,
then $(\arg \lambda -\arg \overline{\lambda})/\pi \in \Integers$.
Thus applying (\ref{eq:lemmap}) and (\ref{eq:diffequality}) finishes the proof.
\end{proof}

The following proposition extends the main result of \cite{kl-zw} and
is very closely related to 
Theorem 5.1 of \cite{bo-pe}.  
Below we use the notation $R_{P+V}(\lambda)$ for the meromorphic continuation
of the resolvent of $P+V$ to $\Lambda$.
\begin{prop}\label{prop:generic} Suppose $P$ satisfies the black box conditions 
recalled in Section \ref{s:correction}, and 
let $U\subset \Real^d$ be the bounded open set as in the statement of 
the hypotheses on $P$,
 with $Pf=-\Delta f$ for $f\in H^2(\Real^d \setminus U)$.
 Let $\chi\in C_c^{\infty}(\Real^d)$ satisfy $\chi \equiv 1$ on $\overline{U}$.
  Then the set of potentials $V\in C_c^{\infty}(\Real^d \setminus \overline{U};\Real)$ 
for which all poles  
of $(1-\chi)R_{P+V}$ are simple with rank $1$ residues
 is a dense $G_\delta$ set in $C_c^{\infty}(\Real^d \setminus \overline{U};\Real)$.
\end{prop}
In fact, \cite{kl-zw} proved that generically the resonances in $\Lambda_1\cup \Lambda_{-1}$ are
simple.  Here we say a statement holds generically if it holds for a dense
$G_\delta $ set.  The proof of \cite{kl-zw} uses complex scaling and so does not obviously immediately extend to other sheets of 
$\Lambda$.  The paper \cite{bo-pe} proved the 
generic simplicity of resonances for $\Delta_g +V$ on 
 asymptotically hyperbolic 
manifolds.  There the main ingredient of the proof 
is Agmon's perturbation theory for resonances \cite{agmon}, 
which is sufficiently general to be applicable to the 
situation described in Proposition \ref{prop:generic}.  
A second part of the proof is a unique continuation theorem,
also valid here.  The proof of Proposition \ref{prop:generic}
 in this setting follows so closely the proof of 
\cite[Theorem 5.1]{bo-pe} that we do not repeat it here.

\vspace{2mm}
\noindent
{\em Proof of Theorem \ref{thm:prpsm}}.
Let $\gamma_{\lambda_0}$ be a small circle centered at $\lambda_0$ 
which does not enclose any other poles of $R_P(\lambda)$
 and so that $\gamma_{\lambda_0}$ encloses no poles of 
$R^*_P(\overline{\lambda})$
except possibly one at $\lambda_0$.  Moreover, we require that both
$R_P(\lambda)$ and $R_P^*(\overline{\lambda})$ are holomorphic on $\gamma_{\lambda_0}$ itself.

For $t\in \Real$ and $V\in C_c^{\infty}(\Real^d \setminus \overline{U})$, 
we shall denote by $R_{P+tV}$ the resolvent of $P+tV$ and
 similarly by $S_{P+tV}$ the associated scattering 
matrix.  For any $V\in C_c^{\infty}(\Real^d \setminus \overline{U})$ so that no poles of 
$R_{P+tV}(\lambda)$ cross $\gamma_{\lambda_0}$ for $t\in [0,1]$, the
operator-valued integral
$$\int_{\gamma_{\lambda_0}} R_{P+tV}(\lambda) d\lambda$$
is continuous for $t\in[0,1]$.  Hence for such $V$ and for
all $t\in[0,1]$, the rank of the residue is constant and equal to its value
at $t=0$:
\begin{align}\label{eq:constrank1}
\rank  \int_{\gamma_{\lambda_0}} R_{P+tV}(\lambda) d\lambda & = 
\mu_{R_{P}}(\lambda_0) 
\end{align}
and 
\begin{align}
\rank  \int_{\gamma_{\lambda_0}}(1-\chi) R_{P+tV}(\lambda)
 d\lambda & = 
\mu_{(1-\chi)R_{P}}(\lambda_0).
\end{align}
We note that it follows from Lemma \ref{l:gu-zw}
that
$$\rank \int_{\gamma_{\lambda_0}}R_{P+tV}(\lambda)(1-\chi)= 
\rank \int_{\gamma_{\lambda_0}}(1-\chi)R_{P+tV}(\lambda).$$
Likewise, if no poles of $R^*_{P+tV}(\overline{\lambda})$ cross
 $\gamma_{\lambda_0}$ for $t\in [0,1]$, then for all $t\in [0,1]$
\begin{align}\label{eq:constrank2}
\rank  \int_{\gamma_{\lambda_0}} R^{*}_{P+tV}(\overline{\lambda}) d\lambda  = 
\mu_{R_{P}}(\overline{\lambda_0})
\end{align}
 and
\begin{align}\label{eq:constantranklast}
\rank   \int_{\gamma_{\lambda_0}} R^*_{P+tV}(\overline{\lambda})(1-\chi)
 d\lambda  = 
\mu_{(1-\chi)R_{P}}(\overline{\lambda_0}),
\end{align}
using that the rank of $A^*$ is equal to the rank of $A$.

By Proposition \ref{prop:generic}, we may choose 
$V\in C_c^{\infty}(\Real^d \setminus \overline{U};\Real)$ so that no poles  of either
$R_{P+tV}(\lambda)$ or $R^*_{P+tV}(\overline{\lambda})$
cross $\gamma_{\lambda_0}$ for $t\in [0,1]$ and so that 
$(1-\chi)R_{P+tV}(\lambda)$ 
has simple poles with rank one residues for 
some $t_0\in [0,1]$.  

Then we use  (\ref{eq:constrank1})- (\ref{eq:constantranklast}), (\ref{eq:lemmap}), and
Lemma \ref{l:multdiff1} to show that 
\begin{equation}\label{eq:t0}
\mu_{R_{P}}(\lambda_0)- \mu_{R_P}(\overline{\lambda_0})=
\mu_{(1-\chi)R_P}(\lambda_0)- \mu_{(1-\chi)R_P}(\overline{\lambda_0}) =
 -\msca(\det(S_{P+t_0V}(\lambda)), \lambda_0).
\end{equation}
When $V$ is chosen as above, 
%since the only poles or zeros of 
%$S_{P}(\lambda)$ inside or on $\gamma_{\lambda_0}$ are at $\lambda_0$,
% the argument principle gives
by Lemma \ref{l:wellknown} and using the choice of $\gamma_{\lambda_0}$,
$$\msca(\det(S_{P}(\lambda)), \lambda_0) = \frac{1}{2\pi i}\tr \int_{\gamma_{\lambda_0}}
S^{-1}_{P}(\lambda)S'_{P}(\lambda)d\lambda.$$
Since neither zeros nor poles of $S_{P}(\lambda)$ lie
on $\gamma_{\lambda_0}$ for $t\in [0,1]$, the operator-valued integral 
$$\frac{1}{2\pi i}\tr \int_{\gamma_{\lambda_0}}
S^{-1}_{P+tV}(\lambda)S'_{P+tV}(\lambda)d\lambda$$
is a continuous function of $t\in [0,1]$, and hence, being an integer, is the constant
$\msca(\det(S_{P}(\lambda)), \lambda_0)$.  Combining this observation with (\ref{eq:t0})
proves the theorem.
\qed

\vspace{2mm}

We give a corollary to Theorem \ref{thm:prpsm} which may be helpful 
in studying resonances
on  $\Lambda_m$.
\begin{cor}\label{c:rpsp}
 Under the hypotheses of Theorem \ref{thm:prpsm},
for $\lambda_1 \in \Lambda$, $m\in \Natural$,
\begin{align*}\mu_{(1-\chi)R}(\lambda_1 e^{im\pi})- \mu_{(1-\chi)R}(\lambda_1) 
& = -\msca( \det( mS(\lambda)-(m-1)I), \lambda_1) \\
&  = \mu_{R}(\lambda_1 e^{im\pi})- \mu_{R}(\lambda_1).
\end{align*}
\end{cor}

\begin{proof}
We note first that $R(\lambda)=R^*(e^{i\pi}\overline{\lambda}) $
means that $\mu_R(\lambda)= \mu_R(e^{i\pi}\overline{\lambda}) $, and 
similarly for $\mu_{(1-\chi)R}$.
Using this and applying Theorem \ref{thm:prpsm} gives
$\mu_R(\lambda_0)-\mu_R(e^{i\pi}\lambda_0)= -\msca(\det S, \lambda_0)$
any $\lambda_0\in \Lambda$.  
Repeatedly using this identity with $\lambda_0$ replaced by $\lambda_1$,
$e^{i\pi}\lambda_1$,..., $e^{i(m-1) \pi}\lambda_1$ in turn and adding gives
\begin{multline}
\mu_R(\lambda_1)-\mu_R(e^{i\pi } \lambda_1)+ 
\mu_R(e^{i\pi } \lambda_1)-  \mu_R(e^{i2 \pi } \lambda_1)+ \cdot \cdot \cdot
+ \mu_R(e^{i(m-1)\pi } \lambda_1)- \mu_R(e^{im\pi } \lambda_1)\\
= -\msca (\det S, \lambda_1) -\msca (\det S, e^{i\pi}  \lambda_1) 
-\cdot \cdot \cdot -\msca (\det S, e^{i(m-1)\pi} \lambda_1).
\end{multline}
Applying Proposition \ref{p:togettophys}, we find
$$\mu_R(\lambda_1)- \mu_R(e^{im\pi } \lambda_1) = -
\msca\left(\det(m S(\lambda)- (m-1)I), \lambda_1 \right).$$
The result for $\mu_{R(1-\chi)}$ follows similarly.
\end{proof}

 %It follows from the 
%results of \cite[Proposition 2.1]{pe-zwsc} that $\mu_p(R,k_0) \geq \mu_p (S,k_0)$.  Thus Proposition \ref{p:onedirection}
%finishes the proof.

\section{Purely imaginary poles and (\ref{eq:correct})}\label{s:poles}

The purpose of this section is to prove some consequences of (\ref{eq:correct})
regarding poles of the scattering matrix in even dimensions.  
Among other things, we shall 
 point out the importance of the 
distinction between (\ref{eq:correct}) and (\ref{eq:incorrect}) 
when related to the question of the existence of purely imaginary poles of the
scattering matrix or resolvent.
In this section we again assume $d$ is even.

\begin{thm}\label{thm:pureimag}
Let $\sigma>0$, and denote by $i\sigma$ the point of $\Lambda$ with 
argument $\pi /2$ and norm $\sigma$.  If $S(i\sigma)-I$
has only purely imaginary eigenvalues, then $S(\lambda)$ is analytic
in a neighborhood of $e^{i(m \pi +\pi/2)}\sigma, $
$m\in \Integers \setminus \{ 0\}$.  Moreover, if  $-\sigma^2$ is not an 
eigenvalue of $P$, then $P$ does not have a resolvent resonance at
$e^{i(m \pi +\pi/2)}\sigma.$
\end{thm}
\begin{proof}
We note that if $S(i \sigma)-I$, a compact operator, has only purely
imaginary eigenvalues, then $(m+1)S(i\sigma)-mI= (m+1)(S(i\sigma)-I)+I$
has no nontrivial null space.  This gives
 $\msca(\det( (m+1)S(\lambda)-mI),i\sigma)\leq 0.$
  Thus if $m>1$ by applying Proposition
\ref{p:togettophys} we see that $S$ cannot have a pole at 
$e^{i(m\pi+ \pi/2)}\sigma$.  

Again assuming $m>0$, we note that $S$ has a pole at 
$e^{i(-\pi m +  \pi/2)}\sigma$
if and only if $S$ has a pole at $e^{i \pi} e^{i(\pi m -  \pi/2)}\sigma
= e^{i(\pi m +  \pi/2)}\sigma$, by (\ref{eq:correct}).  But 
we have just shown this is impossible.  Thus $S$ has no poles at 
$e^{i(m\pi +\pi/2)} \sigma$ for any $m\in \Integers \setminus \{0\}$.

If $P$ does not have eigenvalue $-\sigma^2$, then $R(\lambda)$ and 
$\Phi(\lambda)$ are both regular at $i\sigma$.  Thus combining the first part 
of the theorem with Lemma \ref{l:almostonedirection} we see that
$R(\lambda)$ is regular at $e^{i(m\pi+\pi/2)\sigma}$.
\end{proof}

Now we can see immediately why the distinction between (\ref{eq:incorrect})
and (\ref{eq:correct}) is so important here.  If (\ref{eq:incorrect})
were true, in even dimension $d$ we would have that 
$I-S(i\sigma)$ is skew-adjoint, and hence has only imaginary eigenvalues.  
However, since it is rather that $\mcr (I-S(i\sigma))$ which is skew-adjoint,
the question is more subtle.  However, something can still be said.  

\begin{cor}\label{c:generalpureimag} Let $d$ be even.
Suppose $\pot \in L^{\infty}_{0}(\Real^d;\Real)$ and $\mco \subset \Real^d$
is an open bounded set with smooth boundary.  
Let $P$ be the operator $-\Delta +\pot$
on $\Real^d\setminus \overline{\mco}$, with Dirichlet or Neumann boundary conditions.
If $P$ has no negative eigenvalues and if both $\{x\in \Real^d: -x\in \mco \} =\mco$ and 
$\pot (-x)\equiv \pot (x)$,
then $P$ has no resonances with argument
 $\pi/2+ m\pi$, $m\in \Integers \setminus \{0\}$.
\end{cor}
\begin{proof}
This follows immediately from Theorem \ref{thm:pureimag} 
and Theorem \ref{thm:prpsm} when 
combined with Corollary \ref{c:even} which showed that in this 
setting $I-S(i\sigma)$ is skew adjoint.  Note that we do not
specify whether the purely imaginary resonances are resolvent resonances or
scattering resonances, since the theorem and our assumptions guarantee
that there are neither.
\end{proof}

The case of $m=-1$  (and thus also $m=1$)
of the following corollary is proved in 
\cite[Theorem 4.4]{beale}.  The results of \cite{beale}
combined with Theorem \ref{thm:pureimag} immediately give us more.
\begin{cor}\label{c:obstaclepureimag} Let $\mco \subset \Real^d$
be an open bounded set with smooth boundary 
 so that $\Real^d \setminus \overline{ \mco}$ is 
connected.
Let $P$ the the operator $-\Delta$ on $\Real^d \setminus 
\overline{\mco}$ with 
Dirichlet or Neumann boundary conditions.  Then $P$ has no resonances with argument $\pi/2+ m\pi$, $m\in \Integers \setminus \{0\}$.  If $P$ is instead
the operator $-\Delta$ on $\Real^d \setminus 
\overline{\mco}$  satisfying the 
Robin-type boundary condition 
$$f u +\frac{\partial u}{\partial \nu }=0 \; \text{on $\partial(\Real^d
\setminus  \overline{\mco})$}$$
where $\nu$ is the outward normal and $f$ is a non-negative $C^1$ 
function, then $P$ has at most finitely many resonances (resolvent
or scattering) with
argument $\pi/2+ m\pi$ for each  $m\in \Integers \setminus \{0\}$.
\end{cor}
\begin{proof}
We note first the absence of negative eigenvalues in this setting.

From \cite[Theorem 3.5]{beale}, for the Robin boundary condition for
a fixed obstacle there
is a $\sigma_0$
so that  $i^{-(d-1)}(I-S(e^{i\pi/2} \sigma))\mcr$
is a negative operator for  $\sigma \in (\sigma_0,\infty)$.  From \cite{l-p} or 
\cite{beale}, for the Dirichlet
(Neumann) boundary condition,
 $i^{-(d-1)}(I-S(e^{i\pi/2} \sigma))\mcr$
is positive (negative) for any $\sigma>0$.  From the results 
of \cite[Section 4]{l-p}, the eigenvalues of
$(I-S(e^{i\pi/2} \sigma))\mcr$ are purely imaginary, for $\sigma>\sigma_0$
for the Robin case, and all $\sigma>0$ for the other cases. 
  Then
Theorem \ref{thm:pureimag} finishes the proof, since we know from results
of Vodev \cite{vodeveven,vodev2} that there are only finitely many
resonances on the interval $e^{i(\pi/2+m\pi)}(0,\sigma_0)$.
\footnote{ We note that using results of \cite{mu-st} gives
an alternate approach to those of \cite{vodeveven,vodev2} to showing that 
there are only finitely many resonances on $e^{i(\pi/2+m\pi)}(0,\sigma_0)$,
and indeed in a region $\{ \lambda\in \Lambda: -M\leq \arg \lambda
\leq M,\; |\lambda|\leq M\}$ for any finite $M$.}
\end{proof}

\section{Purely imaginary resolvent resonances for fixed sign potentials}\label{s:fixedsign}

In this section we prove Theorem \ref{thm:fixedsign} on resolvent resonances for Schr\"odinger operators with potentials 
$V\in L^{\infty}_{0}(\Real^d)$ with fixed sign. Again in this section we assume $d$ is even. It would be possible to prove
a slightly weaker version of Theorem \ref{thm:fixedsign}
in a manner analogous to the obstacle case, Corollary \ref{c:obstaclepureimag},
invoking some results of \cite{l-p,vasy}.  However, we choose to do this in a somewhat different way relying on the structure of the resolvent.
This method has some advantages. For example, we
prove that if $e^{i(m\pi+\pi/2)}\sigma$, $\sigma>0$, $m\in \Integers$ is a
resonance of $-\Delta +V$ and $V\leq 0$, $V\in L^{\infty}_{0}(\Real^d)$, then $-\sigma^2$ is an
eigenvalue of $-\Delta +V$.

We use the notation $R_V(\lambda)$ for the meromorphic continuation of the
resolvent of $-\Delta+V$, so that for $\lambda \in \Lambda_0$,
$R_V(\lambda)= (-\Delta +V-\lambda^2)^{-1}$.

As in \cite{ch-hi2}, we reduce the problem to studying an operator on $\Lambda_0$ using the following identity. For $\lambda \in \Lambda_0$,
the point $e^{i m \pi} \lambda \in \Lambda_m$. The resolvent $R_0 (\lambda)$ of $H_0$ satisfies
\begin{equation}\label{eq:reduction1}
R_0 (e^{i m \pi} \lambda) = R_0 (\lambda) + i m  T(\lambda),
\end{equation}
where the operator $T(\lambda)$ has the integral kernel:
\begin{equation}\label{eq:t-op1}
T(x,y;\lambda) =  \frac{1}{2} (2 \pi)^{1-d} \lambda^{d-2} \int_{\Sphere^{d-1}} ~e^{i \lambda \omega \cdot (x-y)} ~d \omega .
\end{equation}
A crucial property of this operator is that for $d \geq 2$  even
 and
$\chi\in L^{\infty}_0(\Real^d;\Real)$, the operator $ \chi T(i \sigma)
\chi$ is self-adjoint for real $\sigma >0$.
To see this, note that from \eqref{eq:t-op1}, we have
\begin{equation}\label{eq:skew1}
T(x,y; i \sigma) = \frac{ (-1)^{d/2 + 1}} {2} (2\pi)^{1-d}\sigma^{d-2}  \int_{\Sphere^{d-1}}
 ~e^{-\sigma \omega \cdot (x-y) } ~d \omega 
\end{equation}
with $d/2\in \Natural$.
It follows, using the change of variable $\omega=-\omega'$,
 that $( \chi T( i \sigma)\chi) ^* = \chi T(i \sigma)\chi$.  We note that
this is in contrast with the case of $d \geq 1$ {\em odd}, where the
difference
$R_0(i\sigma)-R_0(-i\sigma)$ is self-adjoint.

\vspace{3mm}
\noindent
{\it Proof of Theorem \ref{thm:fixedsign}.}
 As outlined above, we
 look for zeros of the form $ \sigma_m = e^{i (m + 1/2) \pi } \sigma$, where $\sigma > 0$ and
$\sigma_m \in \Lambda_m$. We define the multiplication operator $\sgn V(x) = +1$ on $\{ x \in \R^d ~|~ V(x) \geq 0 \}$ and $\sgn V(x) = -1$
for $\{ x \in \R^d ~|~ V(x) < 0 \}$. Then the potential has the decomposition $V(x) = \sgn V(x) ~|V(x)|$.
For $\lambda \in \Lambda_0$ we write the resolvent formula as
\begin{equation}\label{eq:resolvent1}
|V|^{1/2} R_V( \lambda ) |V|^{1/2}( I + ( \sgn V) |V|^{1/2} R_0 ( \lambda ) |V|^{1/2} ) =  |V|^{1/2} R_0( \lambda ) |V|^{1/2},
\end{equation}
so  we study the operator
\begin{equation}\label{eq:K-defn1}
I+ {K}(\lambda) =  I + (\sgn V) |V|^{1/2} R_0(\lambda  ) |V|^{1/2}.
\end{equation}

Since $V$ has compact support, the operator $K(\lambda)$ has an analytic continuation to the 
Riemann surface $\Lambda$. It follows from this fact and
\eqref{eq:resolvent1} that the zeros
of $I+K(\lambda )$ on the $m^{th}$-sheet $\Lambda_m$ are the resonances of the operator $H_V$ on $\Lambda_m$.
%For $\lambda \in \Lambda_0$, the value $e^{i m \pi} \lambda \in \Lambda_m$.
The analytic continuation formula for the free
resolvent (\ref{eq:reduction1}) gives
\begin{align}\label{eq:mresolv1}
K(e^{i \pi m} \lambda) =&  (\sgn V) |V|^{1/2} R_0 ( e^{i m \pi} \lambda ) |V|^{1/2} \nonumber \\
 %=& (\sgn V) |V|^{1/2} ( R_0 ( \lambda ) - m T(\lambda )) |V|^{1/2} \nonumber \\
 =&  (\sgn V) |V|^{1/2} R_0 (  \lambda ) |V|^{1/2}  - i m (\sgn V) |V|^{1/2} T(\lambda) |V|^{1/2} .
\end{align}

%In order to investigate purely imaginary poles on $\Lambda_m$, $m \in \Integers^*$,
%we restrict to $\lambda = i \sigma= e^{i\pi/2}\sigma \in \Lambda_0$, with $\sigma > 0$, so that $e^{i (m \pi + \pi /2)}
%\sigma \in \Lambda_m$ is purely imaginary.
%corresponds to a purely imaginary zero.

%\noindent
%2. Using (\ref{eq:hankel1}), the
%resolvent of the Laplacian may be written as
%\begin{equation}\label{eq:resolv1}
%R_0 ( i \sigma ) = M(\sigma),
%\end{equation}
%where the operator $M(\sigma)$ has the real kernel given by
%\begin{equation}\label{eq:kernel1}
%M(x,y; \sigma) = \frac{1}{2 \pi } \left( \frac{ \sigma}{ 2 \pi | x-y|} \right)^{(d-2) /2} K_{(d-2)/2} %( \sigma
%|x-y|) ,
%\end{equation}
%where the modified Bessel function  $K_\nu ( s)$ is defined in (\ref{eq:hankel1}).
%The representation \eqref{eq:resolv1} is independent of the parity of the dimension $d$.
%Consequently, the operator ${K} (i \sigma) \geq 0$, $\sigma > 0$.

%Since $R_0(i\sigma)=(-\Delta+\sigma^2)^{-1}$, $R_0(i\sigma)$ is
%self-adjoint and $R_0(i\sigma)> 0$.  On the other hand, we showed
%that $i |V|^{1/2}T(i\sigma)|V|^{1/2}$ is skew-adjoint in the case of $d$ even.
%On the other hand, the behavior of the trace class operator $|V|^{1/2} T(\lambda) |V|^{1/2}$
%depends on the parity $d$. It was proved above that $T(i \sigma)$ is self-adjoint when $d$ is odd with a real kernel,
%and skew-adjoint when $d$ is even with a purely imaginary kernel.

%\noindent
%Consequently, the operator $\mathcal{K}_0 (i \sigma) \geq 0$. Let $\sigma_m \equiv e^{i m \pi} i \sigma$.
In order to investigate purely imaginary poles on $\Lambda_m$, $m \in \Integers^*$,
we restrict to $\lambda = i \sigma= e^{i\pi/2}\sigma \in \Lambda_0$, with $\sigma > 0$, so that $e^{i (m \pi + \pi /2)}
\sigma \in \Lambda_m$ is purely imaginary. We obtain from
the reduction in (\ref{eq:mresolv1})
\begin{equation*}
{K} ( i e^{im \pi}\sigma  ) =  (\sgn V) |V|^{1/2} R_0 (i\sigma ) |V|^{1/2}  - i m (\sgn V) |V|^{1/2} T(i \sigma) |V|^{1/2} .
\end{equation*}
Since $R_0(i\sigma)=(-\Delta+\sigma^2)^{-1}$, $R_0(i\sigma)$ is
self-adjoint and $R_0(i\sigma)> 0$.  On the other hand, we showed
that $i |V|^{1/2}T(i\sigma)|V|^{1/2}$ is skew-adjoint in the case of $d$ even.

We now consider the operator $I+K(i e^{im\pi}\sigma)$ for potentials $V$ 
having fixed sign.
If $V \geq 0$, the operator $I +  V^{1/2} R_0(i \sigma ) V^{1/2} $ is strictly positive.
On the other hand,
the trace-class operator
$iV^{1/2} T (i \sigma) V^{1/2}$ is skew-adjoint and therefore has only pure imaginary eigenvalues. Consequently, $I+K(e^{i m \pi}i \sigma)$ has no zeros for $\sigma > 0$,
meaning there are no purely imaginary zeros on any sheet.
If $-V \geq 0$, the formula for $K(i e^{im\pi}\sigma)$ becomes:
\begin{equation*}
{K} ( i e^{im\pi}\sigma ) = - |V|^{1/2} R_0(i\sigma) |V|^{1/2}  + i m |V|^{1/2} T(i \sigma) |V|^{1/2} .
\end{equation*}
The operator $I - |V|^{1/2} R_0 (i\sigma ) |V|^{1/2}$ has a zero if and only if $-\sigma^2$ is an
eigenvalue
of $-\Delta+V$, and the multiplicities agree.   Because the trace-class operator
$i |V|^{1/2} T(i \sigma) |V|^{1/2}$ is still skew-adjoint, there can be no zeros
of $I+K(i e^{i\pi m} \sigma)$ unless $-\sigma^2$ is
an eigenvalue of $-\Delta+V$.
 Consequently, if $N_V$ denotes the 
number of negative eigenvalues of $-\Delta +V$,  there
can be at most $N_V < \infty$ purely imaginary resonances on any 
sheet $\Lambda_m$, $m \in \Integers^*$.
\qed

\end{document}